\documentclass[11pt]{scrartcl}

\usepackage{amsmath}
\usepackage{amssymb}
\usepackage{amsthm}
\usepackage{multirow}
\usepackage[dvipsnames]{xcolor}
\usepackage{framed}
\usepackage{array}

\usepackage{spreadtab}  
\usepackage{enumerate}
\usepackage{natbib}

\newtheorem{theorem}{Theorem}[section]
\newtheorem{lemma}[theorem]{Lemma}
\newtheorem{proposition}[theorem]{Proposition}
\newtheorem{corollary}[theorem]{Corollary}
\newtheorem{conjecture}[theorem]{Conjecture}

\theoremstyle{definition}
\newtheorem{remark}[theorem]{Remark}
\newtheorem{definition}[theorem]{Definition}
\newtheorem{example}[theorem]{Example}

\newcommand{\range}[2]{\in\{#1,\ldots,#2\}}
\newcommand{\brs}[2]{B(#1,#2)}
\newcommand{\crs}[2]{C(#1,#2)}

\newcommand{\maxh}{\textsc{MaxH}}
\newcommand{\uijmax}{u_{ij,\text{max}}}
\newcommand{\prop}[1]{PROP$^*{#1}$}

\title{Democratic Fair Allocation of \\Indivisible Goods\footnote{A preliminary version appeared in Proceedings of the 27th International Joint Conference on Artificial Intelligence~\citep{ijcai2018democratic}. 
This version is substantially different, with several new results (\ref{binary-negative-maxh}, \ref{binary-negative-1outofc}, \ref{binary-positive-1ofbestc-improved} \ref{binary-positive-1ofbestc-identical},
\ref{binary-positive-ef1},
\ref{binary-npcomplete},
\ref{additive-positive-unanimous},
\ref{kgroups-binary-negative-2/3},
\ref{kgroups-binary-negative-maxh},
\ref{kgroups-binary-negative-1outofc},
\ref{kgroups-monotonic-positive-EF},
\ref{kgroups-additive-positive-proportional},
\ref{kgroups-binary-positive-1ofbestc},
\ref{kgroups-positive-1ofbestk}), improved proofs for the old results
(including a bug fix in the proof of Theorem \ref{binary-positive-1-of-c}), and an expanded related work section.}}
\author{
Erel Segal-Halevi\\Ariel University
\and
Warut Suksompong\\University of Oxford
}
\date{\vspace{-3ex}}

\def\headingstyle{\bfseries\large}

\allowdisplaybreaks

\newcommand{\hide}[1]{}

\begin{document}

\maketitle

\begin{abstract}
We study the problem of fairly allocating indivisible goods to groups of agents. Agents in the same group share the same set of goods even though they may have different preferences. Previous work has focused on \emph{unanimous fairness}, in which all agents in each group must agree that their group's share is fair. Under this strict requirement, fair allocations exist only for small groups. We introduce the concept of \emph{democratic fairness}, which aims to satisfy a certain fraction of the agents in each group. This concept is better suited to large groups such as cities or countries. We present protocols for democratic fair allocation among two or more arbitrarily large groups of agents with monotonic, additive, or binary valuations. 
For two groups with arbitrary monotonic valuations, we give an efficient protocol that guarantees envy-freeness up to one good for at least 1/2 of the agents in each group, and prove that the 1/2 fraction is optimal. We also present other protocols that make weaker fairness guarantees to more agents in each group, or to more groups. Our protocols combine techniques from different fields, including combinatorial game theory, cake cutting, and voting.
\end{abstract}

\section{Introduction}
Fair division is the study of how to allocate resources among agents with different preferences so that agents perceive the resulting allocation as fair. This problem occurs in a wide range of situations, from negotiating over international interests and reaching divorce settlements \citep{BramsTa96} to dividing household tasks and sharing apartment rent \citep{GoldmanPr14}, and has been an important topic of study in artificial intelligence, multiagent systems, and computational social choice in the past decade \citep{Endriss10,Thomson16,Markakis17}.  

Two kinds of fairness criteria are common in the literature. The first, \emph{envy-freeness (EF)}, means that each agent finds her share at least as good as the share of any other agent. When allocating indivisible goods, envy-freeness is sometimes unattainable (consider two agents quarreling over a single good), so it is often relaxed to \emph{envy-freeness up to one good (EF1)}, which is always attainable \citep{LiptonMaMo04,Budish11}.
The second kind, \emph{proportionality}, stipulates that each agent should receive $1/n$ of her value for the entire set of goods, where $n$ is the number of agents.
For similar reasons, this cannot always be satisfied when goods are indivisible. A well-studied relaxation is \emph{maximin share fairness}, which means that each agent finds his share at least as good as his \emph{maximin share (MMS)}. This is the best share he can secure by dividing the goods into $n$ parts and getting the worst part. 
Even MMS fairness cannot always be guaranteed, but at least a constant fraction of the MMS can  \citep{ProcacciaWa14}.

Most works on fair division involve individual agents, each of whom has individual preferences. 
Yet, in reality resources often have to be allocated among \emph{groups of agents}, such as families or countries.
A good allocated to a group is shared among the group members, who derive utility from the good. 
For example, when dividing real estate among families, all members of a family enjoy their allocated house and backyard.
In international negotiations, the divided rights and settled outcomes are enjoyed by all citizens of a country. 
When resources are allocated between different departments of a university, all department members benefit from the whiteboards, open space and conference rooms allocated to their department. 
Naturally, different group members may have different preferences. The same share can be perceived as fair by one member and unfair by another member of the same group. Ideally, we would like to find an allocation that is considered fair by \emph{all} agents in \emph{all} groups. However, two recent works show that this ``unanimous fairness'' might be too strong to be practical.

(a) \citet{Suksompong18} shows that when allocating indivisible goods among groups, there might be no allocation that is unanimously EF1. Moreover, there might be no division that gives all agents a positive fraction of their MMS.
This impossibility occurs even for two groups of three agents each.

(b) \citet{SegalhaleviNi15} show that when allocating a divisible good (``cake'') among groups, there might be no division that is unanimously envy-free and gives each group a single connected piece, or even a constant number of connected pieces. In contrast,  with individual agents a connected envy-free division always exists \citep{Stromquist80}.

What do groups do when they cannot attain unanimity? In democratic societies, they use some kind of voting.
The philosophy behind voting is that it is impossible to satisfy everyone, so we should try to satisfy as many members as possible. Based on this observation, we say that a division is  \emph{$h$-democratic fair}, for some fairness notion and for some $h\in[0,1]$, if at least a fraction $h$ of the agents in each group perceive it as fair. In this paper we focus on allocating indivisible goods. We would like $h$, the fraction of $h$appy agents, to be as large as possible. 
We thus pose the following question:
\begin{quote}
	\emph{Given a fairness notion, what is the largest $h$ such that an $h$-democratic fair allocation of indivisible goods can always be found?}
\end{quote}

We study democratic fairness under three different assumptions on the agents' valuations. In the most general case, the agents can have arbitrary \emph{monotonic} valuations on bundles of goods. A more common assumption in the literature is that agents' valuations are \emph{additive}---the value of a bundle is the sum of the values of the goods in the bundle. We also study a special case of additive valuations in which agents' valuations are \emph{binary}---each agent has a set of desired goods and her utility equals the number of desired goods allocated to her group. 

\subsection{Results and techniques}
Naturally, there is a trade-off between the strength of the fairness guarantee and the fraction $h$ of agents to which it can be guaranteed. 
In order to give maximum flexibility to policy-makers, we study different fairness requirements of different strengths.

Initially (\textbf{Section~\ref{sec:2groups-binary}}) we consider two groups with binary agents. 
We study a relaxation of envy-freeness that we call \emph{envy-freeness up to $c$ goods (EF$c$)}, a generalization of EF1. In this case, the trade-off curve is degenerate: for every constant $c$, it is possible to guarantee ${1/2}$-democratic EF$c$, and the fraction $1/2$ is tight.
The same holds for \emph{proportionality up to $c$ goods (\prop{c})}, a relaxation of proportionality that we define in Section~\ref{sec:model}, as well as for MMS fairness. 

To get a more flexible trade-off curve, we study a relaxation of MMS called \emph{$1$-out-of-$c$ MMS}, which is the best share an agent can secure by dividing the goods into $c$ subsets and receiving the worst one. We prove that $1$-out-of-$c$ MMS can be guaranteed to at least $1-1/2^{c-1}$ and at most $1-1/2^c$ of the agents in both groups.
We also study a weaker relaxation called \emph{1-of-best-$c$}, which means that an agent's utility is at least the value of her $c$-th best good.
It can be guaranteed to at least $1-2/(2^c+1)$ and at most $1-1/2^c$ of the agents in both groups.
For $c=2$, we show that 1-of-best-$2$ fairness can be guaranteed to at least $3/5$ and at most $2/3$ of the agents in each group, improving upon the upper bound for general $c$.

These positive results are attained by an efficient round-robin protocol, which we call the \emph{Round-robin with Weighted Approval Voting (RWAV)} protocol, where each group in turn picks a good using weighted approval voting
with carefully calculated weights. 
The weights of agents who lose in early votes are increased in later votes.
We believe this weighted voting scheme can be interesting in its own right as a way to make fair group decisions.

Next \textbf{(Section~\ref{sec:2groups-additive})} we consider two groups whose agents have arbitrary monotonic valuations. 
We present an efficient protocol that guarantees EF1 to at least $1/2$ of the agents in each group (which is tight even for binary agents). When all agents are additive, this protocol guarantees $1/2$ of the MMS to $1/2$ of the agents. This is tight: one cannot guarantee more than $1/2$ of the MMS to more than $1/3$ of the agents.
If we weaken the requirement to 1-of-best-$c$ fairness, we can make this guarantee to a larger fraction of the agents---in particular, the same guarantees from Section \ref{sec:2groups-binary} are valid for additive agents.  

The positive results here use a different protocol, which resembles the well-known ``cut-and-choose'' protocol for dividing a cake between two agents. Despite the simplicity of the protocol, we find the result important since, unlike previous results in this setting \citep{ManurangsiSu17,Suksompong18}, our result holds for worst-case instances with any number of agents in the groups and very general utility functions.

Finally (\textbf{Section~\ref{sec:3groups}}), we 
present three generalizations of our results to $k\geq 3$ groups. 
The first generalization has strong fairness guarantees: when all valuations are monotonic, it guarantees EF2 to $1/k$ of the agents in all groups, and when all valuations are binary, it guarantees both EF1 and MMS to  $1/k$ of the agents in all groups (the factor $1/k$ is tight for EF$c$ for any constant $c$, even when valuations are binary). However, the running time of the protocol might be exponential.

The second generalization uses a polynomial-time protocol but has weaker fairness guarantees: when valuations are additive, it guarantees approximations of proportionality and MMS to $1/k$ of the agents, and when valuations are binary, it guarantees MMS to $1/k$ of the agents. However, it does not guarantee EF1 or any relaxation of envy-freeness.

The third generalization uses a variant of the RWAV protocol to guarantee 1-of-best-$k$ fairness (which implies a positive-factor approximation to the MMS) to at least $1/3$ of the agents in each group, for any number of groups with additive agents.

Some of our results and open questions are summarized in Table~\ref{table:summary}.

\begin{table*}
\newcommand{\yes}[1]{\textcolor{ForestGreen}{Yes (#1)}}
\newcommand{\no}[1]{\textcolor{red}{No (#1)}}
\newcommand{\maybe}{\textcolor{blue}{?}}

\centering
\begin{tabular}
{|c||c|c|c|}
\multicolumn{4}{c}{Bin.+Add. valuations:}
\\
\hline
$h$ 
& \shortstack{Positive-MMS \& \\ 1-of-best-2 fairness}
& $(0,~1/2]$-fraction-MMS 
& $(1/2,~1]$-fraction-MMS  
\\
\hline 
\hline 
$(0,~1/3]$  
& \multirow{3}{*}{\yes{Thm. \ref{binary-positive-1ofbestc-improved}}}
& \multirow{2}{*}{\yes{Cor. \ref{additive-positive-1/2}}}
& \textbf{Bin}: \yes{Cor. \ref{additive-positive-1/2}}, \textbf{Add}: \maybe{}
\\ 
\cline{1-1}
\cline{4-4}
$(1/3,~1/2]$ 
&
&  
&
\shortstack{
\textbf{Bin}: \yes{Cor. \ref{additive-positive-1/2}},~~~~~\\~~~~~ \textbf{Add}: \no{Prop. \ref{additive-negative-1/3}}
}
\\
\cline{1-1}
\cline{3-4}
$(1/2,~3/5]$ 
& 
& \multirow{2}{*}{\maybe{}}
& \multirow{2}{*}{\textbf{Bin}: \maybe{}, \textbf{Add}: \no{Prop. \ref{additive-negative-1/3}}}
\\
\cline{1-1}
\cline{2-2}
$(3/5,~2/3]$ 
& \maybe{} (Conjecture \ref{conj:binary-positive-2/3})
&  
& 
\\ 
\cline{1-4}
$(2/3,~1]$   
& \multicolumn{3}{|c|}{\no{Prop. \ref{binary-negative-2/3}}}
\\ 
\hline 
\end{tabular}
\vspace{3mm}

\begin{tabular}
{|c||c|}
\multicolumn{2}{c}{Bin.+Add.+Mon. valuations:}
\\
\hline
$h$  
& EF$c$ ($\equiv$ \prop{c}) for $c\geq 1$
\\
\hline 
\hline 
$(0,~1/2]$  
& \yes{Thm. \ref{monotonic-positive-1/2}}
\\ 
\hline 
$(1/2,~1]$ 
& \no{Prop. \ref{binary-negative-EFc}}
\\
\hline 
\end{tabular}
~~~~

\vspace{5mm}
\begin{tabular}
{|c||c|c|}
\hline
$h$  
& 
\shortstack{1-of-best-$c$,  $c\geq 3$
\\
Bin.+Add. valuations}

& 
\shortstack{
1-out-of-$c$-MMS, $c\geq 3$
\\
Bin. valuations
}
\\
\hline 
\hline 
$(0,~{1-2/2^c}]$
& \multicolumn{2}{|c|}{\yes{Thm. \ref{binary-positive-1-of-c}}}
\\
\hline 
$({1-2/2^c},~{1-2/(2^c+1)}]$
& \yes{Thm. \ref{binary-positive-1ofbestc-improved}}
& \maybe{}
\\
\hline 
$({1-2/(2^c+1)},~{1-1/ 2^c}]$
& \multicolumn{2}{|c|}{\maybe{}}
\\ 
\hline 
$({1-1/2^c},~1]$ 
& \multicolumn{2}{|c|}{\no{Prop. \ref{binary-negative-1outofc}}}
\\
\hline 
\end{tabular}
\caption{Summary of our results for two groups.
The columns correspond to fairness criteria.
The rows correspond to ranges of $h$ (the fraction of happy agents).
}
\label{table:summary}
\end{table*}

\subsection{Related work}
Fair allocation of indivisible goods is a well-studied topic---see, for example, the survey by \citet{Bouveret2016Fair}.
Perhaps surprisingly, however, the group fair division problem is relatively new. 
We already mentioned the impossibility result of 
\citet{Suksompong18}, which is for worst-case agents' utilities. \citet{KyropoulouSuVo19} explored the possibilities and limitations of using EF1 as the fairness notion in the worst-case setting.
On the other hand, if  the agents' utilities are drawn at random from probability distributions, \citet{ManurangsiSu17} showed that a unanimously envy-free allocation exists with high probability as the number of agents and goods grows.
In our terminology, unanimous fairness is equivalent to \emph{$1$-democratic-fairness}. 
\citet{ghodsi2018rent} studied fair division of rooms and rent among groups of tenants, using three notions of fairness which they term strong, aggregate and weak---their ``strong fairness'' corresponds to our 1-democratic (unanimous) fairness, while their ``weak fairness'' means that at least one agent is satisfied.
A related model, in which a subset of public goods is allocated to a single group of agents but the rest of the goods remain unallocated, has also been investigated \citep{ManurangsiSu18}.


Our group fairness notions differ from those studied, e.g., by 
\citet{Berliant1992Fair}, \citet{Husseinov2011Theory}, \citet{TodoLiHu11}, \citet{benabbou2018diversity} and \citet{ConitzerFrSh19}.
In their setting, goods are divided among \emph{individuals}, each of whom is allocated an individual share. The challenge comes from the requirement to eliminate envy, not only between individuals, but also between subsets of agents.
In our setting, the challenge is that the goods are divided among \emph{groups}, where all members of the same group consume the same share. 
A share that is desirable for some group members might be undesirable for other members of the same group. This motivates the use of social choice techniques such as having each group vote on which goods to pick. 

Group preferences are important in matching markets, too. For example, when matching doctors to hospitals, usually a husband and a wife want to be matched to the same hospital. This issue poses a substantial challenge to stable matching mechanisms \citep{Klaus2005Stable,Klaus2007Paths,Kojima2013Matching}.

The idea of satisfying a certain fairness notion for only a certain fraction of the population, rather than unanimously, is sometimes used when unanimity is provably unattainable. For instance, \citet{ortega2018social} proved that when two matching markets are merged and a stable matching mechanism is run, it is impossible to attain monotonic improvement for everyone, even though the improvement is attainable for at least half of the population.

Our RWAV protocol is reminiscent of  \emph{positional games} \citep{hefetz2014positional}. In a positional game, two players alternately pick elements, trying to occupy an entire set from a given set family, or to occupy a single element in every set of that family. Famous examples are Tic-Tac-Toe and Hex. In our RWAV protocol, the groups are the players, and the sets of elements desired by the group members are the set families. In contrast to positional games, in our case the set family of each group is different, so the ``game'' is not zero-sum. In fact, our goal is to have both groups ``win'' simultaneously to the extent possible.
Our Theorem \ref{binary-positive-1ofbestc}
is equivalent to one of the first results on positional games: the Erd\H{o}s-Selfridge theorem \citep{erdos1973combinatorial}.

Approval voting is widely studied in political science---see, e.g., \citet{brams2018multiwinner} for a recent application of approval voting to multi-winner elections.

\section{Preliminaries}
\label{sec:model}
\subsection{Goods and agents}
There is a set $G=\{g_1,\dots,g_m\}$ of goods.
A \emph{bundle} is a subset of $G$.
There is a set $A$ of agents. 
The agents are partitioned into $k$ groups $A_1,\ldots,A_k$ with $n_1,\ldots,n_k$ agents, respectively. Let $a_{ij}$ denote the $j$th agent in group $A_i$. 

Each agent $a_{ij}$ has a nonnegative utility $u_{ij}(G')$ for each $G'\subseteq G$. 
Denote by $\textbf{u}_{ij}=(u_{ij}(g_1),\dots,u_{ij}(g_m))$ the utility vector of agent $a_{ij}$ for individual goods. 
The agents' utility functions are \emph{monotonic}, i.e., $u_{ij}(G'')\leq u_{ij}(G')$ for every $G''\subseteq G'\subseteq G$ and every agent $a_{ij}$. A subclass of monotonic utilities is the class of \emph{additive} utilities, i.e.,
for every bundle $G'\subseteq G$ and every agent $a_{ij}\in A$, we have $u_{ij}(G')=\sum_{g\in G'} u_{ij}(g)$.
Sometimes we will study a special case of additive utilities in which utilities are \emph{binary}, i.e., each agent either approves or disapproves each good. Since we will not engage in interpersonal comparison of utilities, we may assume without loss of generality that in this case $u_{ij}(g) \in \{0,1\}$ for each $i,j,g$.\footnote{
Binary valuations are similar to \emph{dichotomous preferences}, which are common both in theory \citep{BogomolnaiaMoSt05} and in practice \citep{kurokawa2015leximin}. 
However, in dichotomous preferences each agent assigns a utility of either 0 or 1 to each \emph{bundle} rather than to each good (so they are usually not additive).
}
For brevity, we will refer to agents with monotonic (resp., additive, binary) valuations as \emph{monotonic agents} (resp., \emph{additive agents}, \emph{binary agents}).

We allocate a bundle $G_i\subseteq G$ to each group $A_i$. All goods should be allocated. The goods are treated as public goods within each group, i.e., 
for every group $i$, the utility of every agent $a_{ij}$ is $u_{ij}(G_i)$.
We refer to a setting with agents partitioned into groups, goods and utility functions as an \emph{instance}.

\subsection{Fairness notions}
We begin by defining what it means for an allocation to be fair for a \emph{specific} agent. 
We start with envy-freeness.

\begin{definition}
	\label{def:EF}
	\ifdefined\EFX
	(a)
	\fi
	For an agent $a_{ij}$
	and an integer $c\geq 0$,
	an allocation is called \emph{envy-free up to $c$ goods (EF$c$) for $a_{ij}$}
	if for all $i'$ 
	there is a set 
	$C_{i'}\subseteq G_{i'}$ with $|C_{i'}|\leq c$ such that:
	\begin{align*}
	u_{ij}(G_i)\geq u_{ij}(G_{i'}\backslash C_{i'}).
	\end{align*}
	In other words, one can remove the envy of $a_{ij}$ toward group $i'$ by removing at most $c$ goods from the group's bundle.
	~~An EF0 allocation is also known as \emph{envy-free}.
	
	\ifdefined\EFX
	(b)
	An allocation is called \emph{envy-free up to any good (EFX) for $a_{ij}$}
	if for every $i'$, 
	and for \emph{every}
	set 
	$C_{i'}\subseteq G_{i'}$ with $|C_{i'}| = 1$:
	\begin{align*}
	u_{ij}(G_i)\geq u_{ij}(G_{i'}\backslash C_{i'}).
	\end{align*}
	\fi
\end{definition}

A second common fairness notion is proportionality.
\begin{definition}
For an agent $a_{ij}$ and an integer $c\geq 0$,
an allocation is called \emph{proportional except $c$ goods (\prop{c})} 
for $a_{ij}$
if there exists a set $C\subseteq G\setminus G_i$ with $|C|\leq c$ such that:
\begin{align*}
u_{ij}(G_i)\geq \frac{1}{k} u_{ij}(G\setminus C).
\end{align*}
In other words, 
the agent feels that her family received its proportional share
as long as the best $c$ goods not allocated to it are ignored.
~~

A \prop{0} allocation is also known as \emph{proportional}.
\end{definition}

Envy-freeness is stronger than proportionality in the following sense:
\begin{lemma}
\label{EF1-is-PROP}
If an allocation is EF1 for an additive agent, then it is also \prop{(k-1)} for that agent.
When $k=2$, the opposite implication is also true. 
\end{lemma}

\begin{proof}
Denote by $u$ the utility function of the agent and assume without loss of generality that the agent is in group $A_1$.
EF1 implies that
in each bundle $G_i$ (for $i\range{2}{k}$) there exists a subset $C_i$ with $|C_i|\leq 1$ such that $u(G_1)\geq u(G_i \setminus C_i)$. 
Summing over groups $2,\dots,k$ and adding $u(G_1)$ to both sides gives $
k\cdot u(G_1)
\geq 
u(G 
\setminus (C_2\cup\dots\cup C_k))
$.
Let $C := C_2\cup\dots\cup C_k$. We have $|C|\leq k-1$, $C\subseteq G\setminus G_1$,
and $u(G_1) \geq \frac{1}{k}u(G 
\setminus C)$, so the allocation is \prop{(k-1)} for the agent.%
\footnote{
A similar proof shows that for every $c\geq 1$, 
EF$c$ implies \prop{(c(k-1))}.
}

When $k=2$, \prop{1} implies that $u(G_1)\geq \frac{1}{2}u(G\setminus C)$ for some $C\subseteq G_2$  with $|C|\leq 1$. This implies that 
$u(G_1)\geq u(G\setminus (C\cup G_1)) = u(G_2\setminus C)$, so the allocation is EF1 for the agent.
\end{proof}

We remark that our \prop{(k-1)} is slightly stronger than the PROP1 notion studied by \cite{ConitzerFrSh17}: we remove the $k-1$ highest values in $\{u(g)\mid g\in G\setminus G_i\}$, while they remove $k$ times the highest value in $\{u(g)\mid g\in G\setminus G_i\}$.
Similarly, it is stronger than previous  approximations of proportionality  considered by \cite{ConitzerFrSh17}, \cite{AzizCaIg18}, and \cite{Suksompong17}.

We illustrate our fairness notions with the following example, which also shows that the opposite implication of Lemma \ref{EF1-is-PROP} does not hold when $k>2$.
\begin{example}
Suppose there are $k=3$ groups. Below, all valuations are for a specific agent in group $A_1$.

Consider first an allocation in which group $A_1$ gets a single good worth 30, group $A_2$ gets four goods worth 20, and group $A_3$ gets one good worth 10 (the total value of all goods is 120).
This allocation is \prop{2} for the agent, since if we remove the two most valuable goods not allocated to group $A_1$ (two goods worth 20), the value of the remaining goods is 80, and the value for $G_1$ is more than 80/3.
However, the allocation is not EF1 or even EF2 for the agent, since after removing two goods from $G_2$, the agent's value for the bundle is 40, and her value for $G_1$ is less than 40.
The allocation is also not \prop{1} for the agent, since removing any single good not allocated to $A_1$ leaves a value of at least 100, and the agent's value for $G_1$ is less than 100/3.

Consider an alternative allocation in which one good is taken from $A_2$ and given to $A_1$. In the new allocation, the agent's value for $G_1$ is 50 and her value for $G_2$ is 60.
The allocation is \prop{0} and EF1 for the agent. However, it is still not EF0 (envy-free).
\qed
\end{example}

Next, we define the maximin share and its relaxations.
\begin{definition}
\label{def:maximin}
Given an agent $a_{ij}$, the \emph{maximin share (MMS)} of $a_{ij}$ is the maximum, over all partitions of $G$ into $k$ sets, of the minimum of the agent's utilities for the sets in the partition (below, $(G_1',\dots,G_k')$ is any $k$-partition of $G$):
\[
\text{MMS}^k_{ij}(G) := 
\max_{G_1',\dots,G_k'}\min(u_{ij}(G'_1),\ldots,u_{ij}(G'_k))
\]
An allocation $(G_1,\ldots,G_k)$ is said to be \emph{MMS-fair for $a_{ij}$} if $u_{ij}(G_i)\geq \text{MMS}^k_{ij}(G)$.
\end{definition}
Since MMS fairness is not always attainable, various approximations have been studied. We define, for each integer $c\geq k$, the 
\emph{1-out-of-$c$ maximin share} of $a_{ij}$ as the maximum, over all partitions of $G$ into $c$ sets, of the minimum of the agent's utilities for the sets in the partition (below, $(G_1',\dots,G_c')$ is any $c$-partition of $G$):
\[
\text{MMS}^c_{ij}(G) := 
\max_{G_1',\dots,G_c'}\min(u_{ij}(G'_1),\ldots,u_{ij}(G'_c))
\]
\begin{definition}[relaxations of MMS]
An allocation $(G_1,\ldots,G_k)$ is said to be:
\begin{itemize} 
\item \emph{1-out-of-$c$ MMS-fair for $a_{ij}$}, if $u_{ij}(G_i)\geq \text{MMS}^c_{ij}(G)$.
\item \emph{$q$-fraction-MMS-fair for $a_{ij}$}, for some fraction $q\in(0,1)$, if $u_{ij}(G_i)\geq q\cdot \text{MMS}^k_{ij}(G)$.
\item \emph{1-of-best-$c$ for $a_{ij}$}, for an integer $c\geq k$, if
$u_{ij}(G_i)$ is at least the value of $a_{ij}$'s $c$-th best good in $G$.
\item \emph{positive-MMS-fair for $a_{ij}$}, if $\text{MMS}^k_{ij}(G)>0$ implies $u_{ij}(G_i)>0$ (equivalently, if it is $q$-fraction-MMS-fair for some $q>0$). 
\end{itemize}
\end{definition}

MMS fairness was introduced by \citet{Budish11} based on earlier concepts by \citet{Moulin90}. 
Budish also considered its relaxation to $1$-out-of-$(n+1)$ MMS.
The notion $1$-out-of-$c$ MMS is a special case of $l$-out-of-$d$ MMS, recently defined by \citet{babaioff2017competitive} and studied by \citet{SegalHalevi2018CEFAI}.  
Fractional MMS fairness was introduced by \citet{ProcacciaWa14}.

$1$-of-best-$c$ fairness has not been considered before, as far as we are aware. It is relevant for agents who are mainly interested in getting one of their top choices. While it is a rather weak fairness notion, in our group setting it is useful when stronger fairness notions are unattainable.

The following implications between the MMS approximations follow from their definitions (note that 1-of-best-$k$-fair also implies 1-of-best-$c$-fair for $c\geq k$):
\begin{align*}
& \nearrow \text{1-out-of-$c$ MMS-fair ($c\geq k$)}
\rightarrow \text{1-of-best-$c$-fair ($c\geq k$)}&
\\
\text{MMS-fair} & \rightarrow  \text{1-of-best-$k$-fair} \rightarrow \text{positive-MMS-fair} 
\\
& \searrow \text{$q$-fraction-MMS-fair ($q\in(0,1)$)}  \nearrow
\end{align*}


Proportionality is stronger than MMS in the following sense:

\begin{lemma}
\label{PROP-is-MMS}
If an allocation among $k$ groups is \prop{(k-1)} for an additive agent, then:
	
(a) It is also $1/k$-fraction-MMS-fair for that agent ($1/k$ is tight);%
\footnote{
The implication EF1 $\to$ $1/k$-fraction-MMS-fairness was proved concurrently and independently by \citet[Prop.~3.6]{amanatidis2018comparing}.
}
	
(b) It is also $1$-out-of-$(2 k-1)$ MMS-fair for that agent ($2 k-1$ is tight);
	
(c) If the agent is binary, the allocation is also MMS-fair for that agent. Moreover, MMS-fairness implies \prop{(k-1)} for a binary agent.
\end{lemma}
\begin{proof}
Denote by $u$ the utility function of the agent and assume without loss of generality that the agent is in group $A_1$.
\begin{enumerate}[(a)]
\item \prop{(k-1)} implies that
there exists a subset $C$ with $|C|\leq k-1$ 
such that $u(G_1)\geq \frac{1}{k}u(G \setminus C)$. 
In any partition of $G$ into $k$ bundles, 
at least one bundle
does not contain any good in $C$. 
This bundle is contained in 
$G \setminus C$.
Therefore, the MMS is at most 
$u(G 
\setminus C)$
and $u(G_1)$ is at least $1/k$ of it.

To show that the factor $1/k$ is tight, assume that there are $2k-1$ goods with $u(g_1) = \cdots = u(g_k) = 1$ and $u(g_{k+1}) = \cdots = u(g_{2k-1}) = k$. 
If the agent's group gets $g_1$ and group $i\geq 2$ gets $\{g_i, g_{k+i-1}\}$, the agent gets utility $1$ and finds the allocation EF1 (hence also \prop{(k-1)}). However, the MMS is $k$, as can be seen from the partition $(\{g_1, \ldots, g_k\}, \{g_{k+1}\}, \ldots, \{g_{2k-1}\})$.

\item As above \prop{(k-1)} implies $u(G_1)
\geq \frac{1}{k} u(G \setminus C)$.
Since $|C|\leq k-1$, 
in any partition of $G$ into $2 k-1$ bundles, 
at least $k$ bundles
do not contain any good in $C$.
The union of these bundles is contained in 
$G \setminus C$.
By the pigeonhole principle, at least one of these $k$ bundles has utility at most 
$\frac{1}{k}u(G \setminus C)$
which is at most $u(G_1)$. 
Therefore, $u(G_1)$ is at least the 1-out-of-$(2k-1)$ MMS.%
\footnote{
A similar proof shows that for every $c\geq 1$, 
\prop{c} implies 1-out-of-$(c+k)$-MMS fairness.
}

To show that $2k-1$ is tight,  consider the following allocation:
\begin{itemize}
\item Group $A_1$ gets a single good worth $k-1$;
\item Each group $A_2,\ldots,A_k$ gets one good worth $k$ plus $k-1$ goods worth $1$ each for our distinguished agent in group $A_1$.
\end{itemize}
The allocation is EF1 (hence \prop{(k-1)}) for the agent.
However, the 1-out-of-$(2k-2)$ MMS of the agent is $k$ due to the following partition (note that in total, there are $1$ good worth $k-1$, $k-1$ goods worth $k$, and $(k-1)(k-1) = k^2-2k+1$ goods worth $1$):
\begin{itemize}
\item Bundle 1 has the good worth $k-1$ plus one good worth $1$;
\item Each of the $k-1$ bundles $2,\ldots,k$ has one good worth $k$;
\item Each of the $k-2$ bundles $k+1,\ldots,2k-2$ has $k$ goods worth $1$.

\end{itemize}
\item Suppose the agent has $pk + q$ desired goods, for some integers $p\geq 0$ and $q\range{0}{k-1}$.
Both \prop{(k-1)} and MMS-fairness require that the agent's group receives at least $p$ of these goods, so the conditions are equivalent.
\qedhere
\end{enumerate}
\end{proof}

Unlike (c), the opposite implications of (a) and (b) do not hold even when $k=2$ and agents are binary.
\begin{example}
Assume that there are $k=2$ groups and 6 goods for which all agents have value $1$.
Then, an allocation is 1-out-of-$3$ MMS-fair for an agent iff the agent's value is at least $2$, and it is $1/2$-fraction-MMS fair for all agents iff each agent's value is at least $3/2=1.5$; 
both these conditions can be satisfied for all agents by giving each group at least two goods.
However, an allocation is \prop{1} for an agent iff the agent's value is at least $5/2=2.5$. Hence, the only way to make the allocation \prop{1} for all agents is by giving each group exactly three goods.
\qed
\end{example}

The following diagram summarizes Lemmas \ref{EF1-is-PROP} and \ref{PROP-is-MMS}. 
\begin{align*}
            & \nearrow \text{${1/k}$-fraction-MMS-fair}
\\
\text{EF1} \rightarrow
\text{\prop{(k-1)}} & \leftrightarrow  \text{MMS-fair (for binary agents)}
\\
            & \searrow \text{1-out-of-$(2k-1)$ MMS-fair} 
\end{align*}

Now we are ready to define our main group fairness notion:
\begin{definition}
	\label{def:democratic-fairness}
	For any given fairness notion,
	an allocation $(G_1,\ldots,G_k)$ is said to be \emph{$h$-democratic fair} if it is fair for at least $h\cdot n_i$ agents in group $A_i$, for all $i\range{1}{k}$.
\end{definition}
We also refer to $1$-democratic fairness as \emph{unanimous fairness}.

\section{Two Groups with Binary Valuations}
\label{sec:2groups-binary}
This section considers the setting where there are two groups, the agents have additive valuations, and
each agent either \emph{desires} a good (in which case her utility for the good is $1$) or does not desire it (in which case her utility is $0$). 

\subsection{Negative results}
\label{sub:binary-negative}
Even in the special case of binary valuations, some fairness guarantees are unattainable.
\begin{proposition}
	\label{binary-negative-2/3}
	For any $h>2/3$, there is a binary instance in which no allocation is $h$-democratic positive-MMS-fair.
\end{proposition}
\begin{proof}
Suppose that there are three goods. Each group consists of three members, each of whom has utility 0 for a unique distinct good and utility 1 for each of the other two goods. Each agent has a positive MMS ($1$), but no allocation gives all agents a positive utility.%
\footnote{
This negative result extends  to groups with $3 l$ members for every integer $l\geq 1$. It does not hold when the number of agents is not a multiple of 3. However, we do not address such cases in this paper, as our goal is to develop division algorithms that hold regardless of the number of agents.
}
\end{proof}

We now proceed to prove more general negative results. 
In binary instances, each agent can be represented by two integers, which we denote here by $r$ and $s$:
\begin{itemize}
	\item $r$ is the number of goods that the agent finds desirable;
	\item $s$ is the number of desirable goods that the agent needs to get so that the allocation is considered fair for that agent.
\end{itemize}
Specific fairness requirements define $s$ as a function of $r$. For example, with two groups:
\begin{itemize}
	\item EF$c$ and \prop{c} both mean that $s = \lfloor (r - c+1)/2 \rfloor$. In particular, EF1 and \prop{1} mean that $s = \lfloor r/2 \rfloor$.
	\item 1-out-of-$c$ MMS means $s = \lfloor r/c \rfloor$. So MMS-fairness is equivalent to EF1 and \prop{1}.
	\item 1-of-best-$c$ fairness means that $s = 1$ whenever $r\geq c$.	
	\item Positive-MMS means that $s = 1$ whenever $r\geq 2$ (equivalent to 1-of-best-2).
\end{itemize}
We prove a negative result for general $r$ and $s$, and then use it to derive negative results for specific fairness requirements.
\begin{proposition}
\label{binary-negative-maxh}
Let $r, s$ be integers such that $r\geq s \geq 1$, and let
\begin{align*}
\maxh(r,s) = 
\begin{cases}
0 &  \text{~~when~~} r \leq 2 s - 1;
\\
\frac{1}{2^r}\sum_{i=s}^r \binom{r}{i}
& \text{~~when~~} r\geq 2 s.
\end{cases}
\end{align*}
For any $h>\maxh(r,s)$, there exists a binary instance with two groups in which no allocation is $h$-democratic fair, where each agent desires exactly $r$ goods and needs $s$ desirable goods in order to consider the allocation fair.
\end{proposition}
\begin{proof}
For the case $r\leq 2 s - 1$, consider an instance with $r$ goods, where all agents in both groups desire all goods and needs $s$ desirable goods to be happy. At least one group will get at most $r/2<s$ goods, so all of its members will be unhappy.
	
For the case $r\geq 2 s$, consider an instance with $2 m$ goods, for some $m\gg r$. In each group there are $\binom{2m}{r}$ members, each of whom wants a distinct subset of $r$ goods and needs $s$ desirable goods to be happy. At least one group will get at most $m$ goods. We may assume that the group receives exactly $m$ goods; if the group receives fewer than $m$ goods, the fraction of happy agents can only decrease. In this group, the fraction of happy agents will be at most:
	\begin{align*}
	\frac{
		\sum_{i=s}^r \binom{m}{i}\cdot \binom{m}{r-i}
	}{
		\binom{2m}{r}
	}.
	\end{align*}
	When $m\gg r$, the numerator is approximately $\sum_{i=s}^r \frac{m^i}{i!}\cdot \frac{m^{r-i}}{(r-i)!}
	=
	\sum_{i=s}^r \frac{m^r}{i! (r-i)!}
	$ and the denominator is approximately $\frac{(2 m)^r}{ r!} = \frac{2^r m^r}{r!}$. Therefore when $m\to\infty$ the expression approaches $\frac{1}{ 2^r}\sum_{i=s}^r \binom{r}{i}$%
	.
\end{proof}
For illustration, some values of $\maxh(r,s)$ are shown in Table \ref{tab:maxh}.

\begin{table}
	\begin{center}
		\STautoround{3}
		\begin{spreadtab}{{tabular}{>{\headingstyle}c|cccccccccc}}
			@ $r\downarrow ~ s \implies$
			& \textbf{:=0}        & \textbf{:={[-1,0]+1}}   & \textbf{:={[-1,0]+1}} & \textbf{:={[-1,0]+1}}  & \textbf{:={[-1,0]+1}}  & \textbf{:={[-1,0]+1}}  & \textbf{:={[-1,0]+1}}  & \textbf{:={[-1,0]+1}}  & \textbf{:={[-1,0]+1}}  & \textbf{:={[-1,0]+1}}
			\\ \hline
			0
			& 1          & 0                   & 0        & 0 & 0 & 0 & 0 & 0 & 0 & 0
			\\ 
			[0,-1]+1   
			& 1          & 0   & ([-1,-1]+[0,-1])/2   & ([-1,-1]+[0,-1])/2   & ([-1,-1]+[0,-1])/2   & ([-1,-1]+[0,-1])/2   & ([-1,-1]+[0,-1])/2   & ([-1,-1]+[0,-1])/2   & ([-1,-1]+[0,-1])/2   & ([-1,-1]+[0,-1])/2   
			\\ 
			[0,-1]+1   
			& 1          & \textit{:={([-1,-1]+0.5)/2}}  & 0   & ([-1,-1]+[0,-1])/2   & ([-1,-1]+[0,-1])/2   & ([-1,-1]+[0,-1])/2   & ([-1,-1]+[0,-1])/2   & ([-1,-1]+[0,-1])/2   & ([-1,-1]+[0,-1])/2   & ([-1,-1]+[0,-1])/2   
			\\ 
			[0,-1]+1   
			& 1          & ([-1,-1]+[0,-1])/2   & 0   & ([-1,-1]+[0,-1])/2   & ([-1,-1]+[0,-1])/2   & ([-1,-1]+[0,-1])/2   & ([-1,-1]+[0,-1])/2   & ([-1,-1]+[0,-1])/2   & ([-1,-1]+[0,-1])/2   & ([-1,-1]+[0,-1])/2   
			\\ 
			[0,-1]+1   
			& 1          & ([-1,-1]+[0,-1])/2   & ([-1,-1]+0.5)/2   & ([-1,-1]+[0,-1])/2   & ([-1,-1]+[0,-1])/2   & ([-1,-1]+[0,-1])/2   & ([-1,-1]+[0,-1])/2   & ([-1,-1]+[0,-1])/2   & ([-1,-1]+[0,-1])/2   & ([-1,-1]+[0,-1])/2   
			\\
			[0,-1]+1   
			& 1          & ([-1,-1]+[0,-1])/2   & ([-1,-1]+[0,-1])/2   & 0   & ([-1,-1]+[0,-1])/2   & ([-1,-1]+[0,-1])/2   & ([-1,-1]+[0,-1])/2   & ([-1,-1]+[0,-1])/2   & ([-1,-1]+[0,-1])/2   & ([-1,-1]+[0,-1])/2   
			\\ 
			[0,-1]+1   
			& 1          & ([-1,-1]+[0,-1])/2   & ([-1,-1]+[0,-1])/2   & ([-1,-1]+0.5)/2   & ([-1,-1]+[0,-1])/2   & ([-1,-1]+[0,-1])/2   & ([-1,-1]+[0,-1])/2   & ([-1,-1]+[0,-1])/2   & ([-1,-1]+[0,-1])/2   & ([-1,-1]+[0,-1])/2   
			\\
			[0,-1]+1   
			& 1          & ([-1,-1]+[0,-1])/2   & ([-1,-1]+[0,-1])/2   & ([-1,-1]+[0,-1])/2   & 0   & ([-1,-1]+[0,-1])/2   & ([-1,-1]+[0,-1])/2   & ([-1,-1]+[0,-1])/2   & ([-1,-1]+[0,-1])/2   & ([-1,-1]+[0,-1])/2   
			\\ 
			[0,-1]+1   
			& 1          & ([-1,-1]+[0,-1])/2   & ([-1,-1]+[0,-1])/2   & ([-1,-1]+[0,-1])/2   & ([-1,-1]+0.5)/2   & ([-1,-1]+[0,-1])/2   & ([-1,-1]+[0,-1])/2   & ([-1,-1]+[0,-1])/2   & ([-1,-1]+[0,-1])/2   & ([-1,-1]+[0,-1])/2   
			\\ 
			[0,-1]+1   
			& 1          & ([-1,-1]+[0,-1])/2   & ([-1,-1]+[0,-1])/2   & ([-1,-1]+[0,-1])/2   & ([-1,-1]+[0,-1])/2   & 0   & ([-1,-1]+[0,-1])/2   & ([-1,-1]+[0,-1])/2   & ([-1,-1]+[0,-1])/2   & ([-1,-1]+[0,-1])/2   
			\\
			[0,-1]+1   
			& 1          & ([-1,-1]+[0,-1])/2   & ([-1,-1]+[0,-1])/2   & ([-1,-1]+[0,-1])/2   & ([-1,-1]+[0,-1])/2   & ([-1,-1]+0.5)/2   & ([-1,-1]+[0,-1])/2   & ([-1,-1]+[0,-1])/2   & ([-1,-1]+[0,-1])/2   & ([-1,-1]+[0,-1])/2   
			\\
		\end{spreadtab}
	\end{center}
	\caption{\label{tab:maxh}Some values of $\maxh(r,s)$. 
		The italics at $(r,s)=(2,1)$ denotes that this upper bound is not tight---Proposition \ref{binary-negative-2/3} shows an upper bound of $2/3$ in this case. We do not know if the other bounds are tight.
		\ifdefined\propertyB
		The bound at (3,1) is probably not tight too. There is a family of 7 sets with 3 elements without property B: C = {{1, 2, 4}, {2, 3, 5}, {3, 4, 6}, {4, 5, 7}, {5, 6, 1}, {6, 7, 2}, {7, 1, 3}}.
		This probably implies an upper bound of $6/7$.
		Property B does not help us anymore when $r\geq 4$: every family of sets with 4 elements without property B provably has at least 23 sets, which implies an upper bound of $22/23$.
		\fi
	}
\end{table}

\begin{proposition}
\label{binary-negative-1outofc}
For any integer $c\geq 2$ and 
$h>1-1/2^c$, there is a binary instance with two groups in which no allocation is $h$-democratic 
$1$-of-best-$c$ fair 
(hence no allocation is 
$h$-democratic $1$-out-of-$c$ MMS-fair).
\end{proposition}

\begin{proof}
	Apply Proposition~\ref{binary-negative-maxh} with $r=c$ and $s=1$.
	Then $\maxh(r,s) = 1 - \frac{1}{2^c}\cdot 1$.
\end{proof}
Note that for $c=2$, 1-of-best-$c$ is equivalent to positive-MMS, so  Proposition~\ref{binary-negative-2/3} gives a tighter upper bound of $2/3$.

\begin{proposition}
	\label{binary-negative-EFc}
	For any constant integer $c\geq 1$ and $h>1/2$, there is a binary instance with two groups in which no allocation is $h$-democratic \prop{c} / EF$c$.
\end{proposition}
\begin{proof}
	Let $l$ be a large positive integer. 
	Apply Proposition~\ref{binary-negative-maxh} with $r = 2 l$
	and $s=\lfloor(r-c+1)/2\rfloor = l - \lfloor c/2 \rfloor$.
	Then $\maxh(r,s) = \frac{1}{2^{2 l}} \sum_{i=l-\lfloor c/2 \rfloor}^{2 l} \binom{2 l}{ i}$.
	When $l\to \infty$, this expression approaches $1/2$ for any constant $c$.
\end{proof}

\ifdefined\EFX
We have an upper bound of $1/2$ for EF1 and a (trivial) upper bound of $0$ for EF.
Interestingly, at the midpoint between these bounds we have an upper bound for EFX.
\begin{proposition}
	\label{binary-negative-EFx}
	For any $h>1/4$, there is a binary instance with two groups in which no allocation is $h$-democratic EFX.
\end{proposition}
\begin{proof}
	Consider an instance with $m=2 l + 1$ goods and $\binom{2 l + 1}{2 l - 1}$ agents in each group, for some $l\geq 1$. Each agent desires a unique subset of $2 l - 1$ goods.
	An allocation is EFX for an agent iff either her group receives at least $l$ desired goods,
	or her group receives at least $l-1$ desired goods and the $2$ non-desired goods.
	However, in any allocation at least one group receives (at most) $l$ goods. 
	In this group, the second condition cannot be satisfied since it requires $l+1$ goods,
	and the first condition is satisfied for $\binom{l+1}{l-1}$ agents (one agent for each choice of $l-1$ desired goods from the set of $l+1$ goods given to the other group). Therefore the fraction of happy agents is:
	\begin{align*}
	\frac{
		\binom{l+1}{l-1}
	}{
		\binom{2 l + 1}{2 l - 1}
	}
	=
	\frac{
		(l+1)l/2
	}{
		(2l +1)(2l)/2
	}.
	\end{align*}
	When $l\to\infty$, this fraction approaches $1/4$.
\end{proof}
\fi

\subsection{Positive results}
\label{sub:binary-positive}
\subsubsection{The RWAV protocol}
Our positive results are attained with a protocol that we call \emph{Round-robin with Weighted Approval Voting (RWAV)}.
The top-level protocol is \emph{round-robin}---the groups take turns picking one good at a time. Formally:
\begin{framed}
\noindent
While there are remaining goods:
\begin{itemize}
\item Group 1 picks a good;
\item Group 2 picks a good.
\end{itemize}
\end{framed}
To attain the fairness guarantee, each group in its turn should pick its good using \emph{weighted approval voting}.
It uses a weight function $w: {\mathbb{Z}_{\geq 0}}\times {\mathbb{Z}_{\geq 0}}\to [0,1]$, which will be specified later using a recurrence relation (see Table \ref{tab:B}/bottom for some example values).
The group assigns to each member $j$ a weight $w(r_j,s_j)$, where:
\begin{itemize}
\item $r_j$ is $j$'s value for the remaining goods (the number of remaining goods that $j$ values at $1$);
\item $s_j$ depends on the chosen fairness criterion: it is the value that $j$ should still get in order to feel that the allocation is fair 
(the number of goods, from the set of remaining goods that $j$ values as 1, that its group should take so that the  fairness criterion is fulfilled for $j$). 
\end{itemize}

Using the members' weights, the group conducts an ``approval voting'' among the remaining goods: each member votes for all the remaining goods he/she values as 1. The total-weight of each good $g$ is the sum of weights of all  agents voting for $g$. The group picks a good $g$ with a maximum total-weight, breaking ties arbitrarily. 
To summarize, here is the algorithm by which each group in its turn should pick a good:%
\footnote{We are grateful to an anonymous reviewer for the suggestion to simplify the protocol.}
\begin{framed}
\noindent
\textbf{Group's strategy for picking a good}:
\begin{itemize}
\item Assign a weight $w(r,s)$ to each group member who wants $r$ of the remaining goods and needs $s$ of them to achieve fairness;
\item For each remaining good $g$, calculate the sum of weights of all members who want $g$, and pick a good with a maximum total-weight.
\end{itemize}
\end{framed}

\begin{example}
\label{exm:rwav}
Consider an instance with five goods, $\{v,w,x,y,z\}$,
where the fairness criterion is 1-out-of-$2$ MMS-fairness.
Suppose that in the first turn, group 1 took $v$ and group 2 took $w$.
Now it is again the turn of group 1.
Assume that in group 1 there are two agents: Alice, whose desired goods are $\{w,x\}$, and Bob, whose desired goods are $\{w,x,y,z\}$:
\begin{itemize}
\item  Alice has a single remaining desired good ($x$), her current value of group 1's share ($v$) is $0$, and she needs a value of $1$ to satisfy the fairness criterion. Therefore $r_{Alice}=|\{y\}| = 1$, $s_{Alice}=1-0=1$, and her weight is $w(1,1)=0.5$ (Table \ref{tab:B}).
\item  Bob has three remaining desired goods ($x,y,z$), his current value of group 1's share is $0$, and he needs a value of $2$ to satisfy the fairness criterion. Therefore $r_{Bob}=3$, $s_{Bob}=2$, and his weight is $w(3,2)=0.375$. \qed
\end{itemize} 
\end{example}

\begin{table}
	\begin{center}
		\STautoround{3}
		\begin{spreadtab}{{tabular}{>{\headingstyle}c|ccccccc}}
			@ $r\downarrow$ $|$ $s \rightarrow$
			& \textbf{:={0}}        & \textbf{:={[-1,0]+1}}   & \textbf{:={[-1,0]+1}} & \textbf{:={[-1,0]+1}}  & \textbf{:={[-1,0]+1}}  & \textbf{:={[-1,0]+1}}  & \textbf{:={[-1,0]+1}} 
			\\ \hline
			0
			& \textbf{:={1}}          & 0                   & 0        & 0 & 0 & 0 & 0 
			\\
			[0,-1]+1   
			& \textbf{:={1}}          & ([-1,-1]+[0,-1])/2   & min([-1,-2],([-1,-1]+[0,-1])/2)   & min([-1,-2],([-1,-1]+[0,-1])/2)   & min([-1,-2],([-1,-1]+[0,-1])/2)   & min([-1,-2],([-1,-1]+[0,-1])/2)   & min([-1,-2],([-1,-1]+[0,-1])/2)   
			\\
			[0,-1]+1   
			& 1       & \textbf{:={min([-1,-2],([-1,-1]+[0,-1])/2)}}   & min([-1,-2],([-1,-1]+[0,-1])/2)   & min([-1,-2],([-1,-1]+[0,-1])/2)   & min([-1,-2],([-1,-1]+[0,-1])/2)   & min([-1,-2],([-1,-1]+[0,-1])/2)   & min([-1,-2],([-1,-1]+[0,-1])/2)   
			\\
			[0,-1]+1   
			& 1          & \textbf{:={min([-1,-2],([-1,-1]+[0,-1])/2)}}   & min([-1,-2],([-1,-1]+[0,-1])/2)   & min([-1,-2],([-1,-1]+[0,-1])/2)   & min([-1,-2],([-1,-1]+[0,-1])/2)   & min([-1,-2],([-1,-1]+[0,-1])/2)   & min([-1,-2],([-1,-1]+[0,-1])/2)   
			\\ 
			[0,-1]+1   
			& 1          & \textbf{:={min([-1,-2],([-1,-1]+[0,-1])/2)}}   & min([-1,-2],([-1,-1]+[0,-1])/2)   & min([-1,-2],([-1,-1]+[0,-1])/2)   & min([-1,-2],([-1,-1]+[0,-1])/2)   & min([-1,-2],([-1,-1]+[0,-1])/2)   & min([-1,-2],([-1,-1]+[0,-1])/2)   
			\\
			[0,-1]+1   
			& 1          & min([-1,-2],([-1,-1]+[0,-1])/2)   & \textbf{:={min([-1,-2],([-1,-1]+[0,-1])/2)}}   & min([-1,-2],([-1,-1]+[0,-1])/2)   & min([-1,-2],([-1,-1]+[0,-1])/2)   & min([-1,-2],([-1,-1]+[0,-1])/2)   & min([-1,-2],([-1,-1]+[0,-1])/2)   
			\\ 
			[0,-1]+1   
			& 1          & min([-1,-2],([-1,-1]+[0,-1])/2)   & \textbf{:={min([-1,-2],([-1,-1]+[0,-1])/2)}}   & min([-1,-2],([-1,-1]+[0,-1])/2)   & min([-1,-2],([-1,-1]+[0,-1])/2)   & min([-1,-2],([-1,-1]+[0,-1])/2)   & min([-1,-2],([-1,-1]+[0,-1])/2)   
			\\ 
			[0,-1]+1   
			& 1          & min([-1,-2],([-1,-1]+[0,-1])/2)   & \textbf{:={min([-1,-2],([-1,-1]+[0,-1])/2)}}   & min([-1,-2],([-1,-1]+[0,-1])/2)   & min([-1,-2],([-1,-1]+[0,-1])/2)   & min([-1,-2],([-1,-1]+[0,-1])/2)   & min([-1,-2],([-1,-1]+[0,-1])/2)   
			\\ 
			[0,-1]+1   
			& 1          & min([-1,-2],([-1,-1]+[0,-1])/2)   & min([-1,-2],([-1,-1]+[0,-1])/2)   & \textbf{:={min([-1,-2],([-1,-1]+[0,-1])/2)}}   & min([-1,-2],([-1,-1]+[0,-1])/2)   & min([-1,-2],([-1,-1]+[0,-1])/2)   & min([-1,-2],([-1,-1]+[0,-1])/2)    
			\\ 
			[0,-1]+1   
			& 1          & min([-1,-2],([-1,-1]+[0,-1])/2)   & min([-1,-2],([-1,-1]+[0,-1])/2)   & \textbf{:={min([-1,-2],([-1,-1]+[0,-1])/2)}}   & min([-1,-2],([-1,-1]+[0,-1])/2)   & min([-1,-2],([-1,-1]+[0,-1])/2)   & min([-1,-2],([-1,-1]+[0,-1])/2)    
			\\ 
			[0,-1]+1   
			& 1          & min([-1,-2],([-1,-1]+[0,-1])/2)   & min([-1,-2],([-1,-1]+[0,-1])/2)   & \textbf{:={min([-1,-2],([-1,-1]+[0,-1])/2)}}   & min([-1,-2],([-1,-1]+[0,-1])/2)   & min([-1,-2],([-1,-1]+[0,-1])/2)   & min([-1,-2],([-1,-1]+[0,-1])/2)   
			\\ 
		\end{spreadtab}
		
		\vspace{2mm}
		$B(r,s)$ for some $r\geq 0$ and $s\geq 0$.
		The boldfaced cells are the cells of $B(r-1, s(r))$
		where $s(r)=\left\lfloor{\frac{r}{3}}\right\rfloor = $ the function corresponding to 1-out-of-3 MMS-fairness.
		~\\
		~\\
		~\\
		\begin{tabular}{>{\headingstyle}c|ccccccc}
			$r \downarrow$ $|$ $s \rightarrow$ &\headingstyle 0 & \headingstyle 1   & \bfseries 2   & \headingstyle 3   & \headingstyle 4   & \headingstyle 5   & \headingstyle 6   \\ 
			\hline 
			0       & 0 & 0 & 0 &  &  &  &  \\ 
			1       & 0 & .500 & 0 &  &  &  &  \\ 
			2       & 0 & .250 & 0\hide{[.500]} & 0 &  &  &  \\ 
			3       & 0 & .125 & .375 & 0 & &  &  \\ 
			4       & 0 & .063 & .250 & 0\hide{[.375]} & 0 & &  \\ 
			5       & 0 & .031 & .156 & .313 &  0 &  & \\ 
			6       & 0 & .016 & .094 & .234 & 0\hide{[.313]} & 0 & \\ 
			7       & 0 & .008 & .055 & .164 & .273 & 0 & \\ 
			8       & 0 & .004 & .031 & .109 & .219 & 0\hide{[.273]} & 0 \\ 
			9       & 0 & .002 & .018 & .070 & .164 & .246 & 0 \\
			10      & 0 & .001 & .010 & .044 & .117 & .205 & 0\hide{[.246]} \\
		\end{tabular}
		\\
		\vspace{2mm}
		$w(r,s)$ for some $r\geq 0$ and $s\geq 0$.
	\end{center}
	\caption{\label{tab:B}
		Some values of $B(r,s)$ and $w(r,s)$. 
		Compare to $\maxh(r,s)$ in Table \ref{tab:maxh}.
	}
\end{table}

We now specify the function $w(r,s)$. We first define an auxiliary function  $B:\mathbb{Z}\times\mathbb{Z}\to [0,1]$ using the recurrence relation below:
\begin{align}
\label{eq:brs}
\brs{r}{s} := 
\begin{cases}
1  &  s\leq 0;
\\
0  &  0<s \text{~and~} r<s;
\\
\min\bigg[
\frac{1}{2}[B(r-1,s)+B(r-1,s-1)]
,
B(r-2,s-1)
\bigg] & \text{otherwise}.
\end{cases}
\end{align}
A closed-form expression of $B(r,s)$ is derived in Appendix~\ref{sec:properties-of-B}.
Table \ref{tab:B} shows some values.

Now $w$ is defined by:
\begin{align}
\label{eq:wrs}
w(r,s) := B(r,s) - B(r-1,s).
\end{align}

Note that $w(r,s)=0$ whenever $s=0$; thus, members for whom the current allocation is already fair do not affect the weighted voting. Similarly, 
$w(r,s)=0$ whenever $r=0$, so members who do not value any of the remaining goods do not affect the voting either.

To demonstrate the operation of the protocol, we have implemented it in Python.%
\footnote{
{https://github.com/erelsgl/family-fair-allocation}
}
An example run is shown in Appendix \ref{sample:rwav}.

\subsubsection{Analysis of RWAV}
RWAV provides a democratic fairness guarantee for various fairness criteria. 
A fairness criterion is represented by an integer function $s(r)$ that maps the total number of desired goods of an agent (which we denoted above by $r$) to the number of desired goods this agent should get in order to satisfy the fairness criterion (which we denoted above by $s$). For example, for 1-out-of-3 MMS-fairness, this function is $s(r) = \lfloor{\frac{r}{3}}\rfloor$.

\begin{lemma}
\label{binary-positive-general}
Given a fairness criterion represented by an integer function $s(r)$, for every group $i\in\{1,2\}$, 
the RWAV protocol 
yields an allocation that is fair for at least 
a fraction $h_i$ of the agents in group $i$, where:
\begin{align*}
h_1 &= \inf_{r = 1, 2, \ldots} B(r,s(r));
\\
h_2 &= \inf_{r = 1, 2, \ldots} B(r-1,s(r)).
\end{align*}
\end{lemma}
\begin{example}
Suppose the chosen fairness criterion is 1-out-of-3 MMS-fairness. Then $h_2$ is the infimum of the sequence $B(r-1, \lfloor{\frac{r}{3}}\rfloor)$, illustrated by the boldfaced cells in Table \ref{tab:B}. We show in Lemma \ref{binary-positive-1ofc} that this infimum is $B(2,1)=0.75$, and in Lemma \ref{B-increasing} that when $s$ is fixed, $B(r,s)$ is an increasing function of $r$. This implies $h_1\geq h_2$ (in this example $h_1=0.875$). Therefore, Lemma \ref{binary-positive-general} implies that the allocation returned by RWAV is $0.75$-democratic 1-out-of-3 MMS-fair.
\qed
\end{example}

To prove Lemma \ref{binary-positive-general}, we add to the protocol several steps in which each group pays/receives fiat money to/from its members. We emphasize that these additional steps are not needed in a practical implementation of the protocol---they are used only in the analysis. We present the additional steps for group 1; the steps for group 2 are analogous.
\begin{itemize}
\item \emph{Initialization}: before the protocol starts, the \emph{balance} of the group and the balance of each member is initialized to 0. Then, each member $j$ pays $B(r_j, s_j)$ to its group, so now $j$'s balance is $-B(r_j,s_j)$ and the group balance is $+\sum_{j\in \text{group 1}} B(r_j,s_j)$.
\item After group 1 picks a good $g$, every member $j$ who values $g$ at 1  (and thus is satisfied with the group's choice) pays $\max[w(r_j,s_j),w(r_j-1,s_j-1)]$ to group 1.%
\footnote{\label{ftn:conference}In the conference version of this paper \citep{ijcai2018democratic}, 
we erroneously wrote that such a member 
should always pay $w(r_j,s_j)$.
}
\item After group 2 picks a good $g$, every member $j$ who values $g$ at 1  (and thus is dissatisfied with the other group's choice) receives $w(r_j,s_j)$ from group 1.
\end{itemize}

\begin{example}[Example \ref{exm:rwav} continued]
Before the protocol starts, Alice has $r=2$ and $s=1$ so she pays $B(2,1)=0.75$, while Bob has $r=4$ and $s=2$ so he pays $B(4,2)=0.625$. The balance of group 1 is $1.375$.

In the first turn, 
Alice's weight is $w(2,1)=0.25$ and Bob's weight is $w(4,2)=0.25$. Goods $w$ and $x$ both have the same total weight of $0.5$, so group 1 picks one of them arbitrarily, say $w$. 
Now, Alice pays $\max[w(2,1),w(1,0)]=\max[0.25,0]=0.25$ 
and Bob pays $\max[w(4,2),w(3,1)]=\max[0.25,0.125]=0.25$, so group 1's balance is now $1.875$.

In the second turn, Alice has $r=1$ and $s=0$ so her weight is $w(1,0)=0$, and Bob has $r=3$ and $s=1$ so his weight is $w(3,1)=0.125$.
Suppose that group 2 picks $x$. Since both Alice and Bob value $x$ at 1, both agents receive their weight from the group---Alice receives $0$ and Bob receives $0.125$. 
Group 1's balance is now $1.75$.

In the third turn, 
Alice's weight is $w(0,0)=0$ and Bob's weight is $w(2,1)=0.25$. Goods $y$ and $z$ both have the same total weight of $0.25$, so group 1 picks one of them arbitrarily, say $y$. 
Now, Alice pays nothing (since $y$ is not one of her desired goods), while Bob pays  $\max[w(2,1),w(1,0)]=\max[0.25,0]=0.25$. Group 1's balance is now $2$, and the balance of both Alice and Bob is $-1$.

From here on, both Alice's and Bob's weights remain 0, so they do not pay nor receive anything and do not affect the group's vote; the group just picks arbitrary goods.

In the final allocation, Alice's value is 1 and Bob's value is 2, so 1-out-of-2-MMS is satisfied for both of them.
\qed
\end{example}

To establish Lemma \ref{binary-positive-general}, we prove some auxiliary lemmas about the behavior of the protocol with the additional steps.
\begin{lemma}
\label{lem:balance}
During the protocol, the balance of each agent $j$ is always $-B(r_j,s_j)$.
\end{lemma}
\begin{proof}
We prove the claim for members of group 1; the proof for group 2 is analogous.

The proof is by induction. The induction  base is handled by the initialization step.

After group 2 picks a good $g$,
the balance of each member who values $g$ at 1 increases.
The new  balance of each such member with parameters $r$ and $s$ is:
\begin{align*}
&
- B(r,s) + w(r,s)
~~=~~
- B(r-1,s) && \text{by definition of $w$,}
\end{align*}
and indeed, for each such member, $r$ drops by 1 while $s$ remains unchanged.
For each member who values $g$ at 0, both the balance and $r,s$ do not change.

After group 1 picks a good $g$, 
the balance of each member who values $g$ at 1 decreases.
The new  balance of each such member with parameters $r$ and $s$ is $- B(r,s) - \max[w(r,s), w(r-1,s-1)]$.
For each such member, both $r$ and $s$ drop by 1, so we have to prove that this new balance equals $-B(r-1,s-1)$. We consider the two cases relevant for the max operation:

\emph{Case 1:} $w(r,s)\geq w(r-1,s-1)$. This implies:
{\small
\begin{align*}
&
B(r,s)-B(r-1,s) \geq B(r-1,s-1)-B(r-2,s-1) ~~~ \text{by definition of $w$}
\\
\implies&
B(r,s)+B(r-2,s-1) \geq B(r-1,s)+B(r-1,s-1) 
\\
\implies&
B(r-2,s-1) \geq \frac{1}{2}[B(r-1,s)+B(r-1,s-1)]
\\
&
\text{~~~~since by \eqref{eq:brs}, either $B(r,s) = B(r-2,s-1)$
or $B(r,s) = \frac{1}{2}[B(r-1,s)+B(r-1,s-1)]$}
\\
\implies&
B(r,s) = \frac{1}{2}[B(r-1,s)+B(r-1,s-1)] ~~~ \text{by \eqref{eq:brs}}
\end{align*}
}
Then the new balance is:
\begin{align*}
&
- B(r,s) - w(r,s)
\\
=& 
- 2 B(r,s) + B(r-1,s) && \text{by definition of $w$}
\\
=& 
- [B(r-1,s)+B(r-1,s-1)] + B(r-1,s) && \text{substituting $B(r,s)$ from above}
\\
=& -B(r-1,s-1).
\end{align*}

\emph{Case 2:} $w(r,s)\leq w(r-1,s-1)$. This implies:
{\small
\begin{align*}
&
B(r,s)-B(r-1,s) \leq B(r-1,s-1)-B(r-2,s-1) ~~~ \text{by definition of $w$}
\\
\implies&
B(r,s)+B(r-2,s-1) \leq B(r-1,s)+B(r-1,s-1) 
\\
\implies&
B(r-2,s-1) \leq \frac{1}{2}[B(r-1,s)+B(r-1,s-1)]
\\
&
\text{~~~~since by \eqref{eq:brs}, either $B(r,s) = B(r-2,s-1)$
or $B(r,s) = \frac{1}{2}[B(r-1,s)+B(r-1,s-1)]$}
\\
\implies&
B(r,s) = B(r-2,s-1) ~~~ \text{by \eqref{eq:brs}}
\end{align*}
}
Then the new balance is:
\begin{align*}
&
- B(r,s) - w(r-1,s-1)
\\
=& 
- B(r,s) - B(r-1,s-1) + B(r-2,s-1) && \text{by definition of $w$}
\\
=& 
- B(r-2,s-1) - B(r-1,s-1) + B(r-2,s-1) && \text{substituting $B(r,s)$ from above}
\\
=& -B(r-1,s-1).
\end{align*}
This completes the proof.
\end{proof}

\begin{lemma}
\label{lem:balance-increasing}
For each group $i$, 
in each pair of consecutive turns in which group $i$ picks a good and then the other group picks a good,
the balance of group $i$ weakly increases.
\end{lemma}
\begin{proof}
	We calculate the change in the balance of group $i$ in a pair of turns in which group $i$ picks a good $g_i$ and then the other group picks a good $g_{-i}$.
	The change in balance is determined by the weights of three groups of agents, which we denote by:
	\begin{itemize}
		\item $D_i$: agents who desire $g_i$ and do not care about $g_{-i}$. Each agent $j$ in this group pays $\max[w(r_j,s_j),w(r_j-1,s_j-1)] \geq w(r_j,s_j)$.
		\item $D_{-i}$: agents who do not care about $g_i$ and desire $g_{-i}$. Each agent in this group receives $w(r_j,s_j)$.
		\item $D_{0}$: agents who desire both $g_i$ and $g_{-i}$. Each agent in this group first pays $\max[w(r_j,s_j), w(r_j-1,s_j-1)]$ and then receives $w(r_j-1,s_j-1)$.
	\end{itemize}
Each agent in $D_0$ pays at least as much as he receives. 	Therefore, the total change in the group balance after the two turns satisfies:
	\begin{align*}
	\Delta[Balance] \geq 
	\sum_{j\in D_i} w(r_j,s_j) - \sum_{j\in D_{-i}} w(r_j,s_j).
	\end{align*}
	Now, the group chose $g_i$ while $g_{-i}$ was still available, which means that the total weight of $g_i$ is weakly larger:
	\begin{align*}
	\sum_{j\in D_i} w(r_j,s_j) + \sum_{j\in D_0} w(r_j,s_j)
	\geq
	\sum_{j\in D_{-i}} w(r_j,s_j) + \sum_{j\in D_0} w(r_j,s_j).
	\end{align*}
	This means that $\Delta[Balance]\geq 0$, as desired.
\end{proof}

Let us call a group member \emph{happy} if the final allocation is fair according to this member's valuation function and the chosen fairness criterion. Otherwise the member is \emph{unhappy}.
\begin{lemma}
\label{lem:rwav-end}
When the RWAV protocol ends, the balance of each group equals the number of its happy members.
\end{lemma}
\begin{proof}
When the protocol ends, all agents have $r=0$. For a happy member $s = 0$, while for an unhappy member $s>0$ so $s>r$. By the recurrence relation of $B$, for a happy member $B(r,s)=B(r,0)=1$, while for an unhappy member $B(r,s)=0$.
By Lemma \ref{lem:balance} the balance of each happy member is $-1$ and the balance of each unhappy member is $0$. 
Since all payments are between the group and its members, the group balance is the negative of the sum of its members' balances, which is exactly the number of happy members.
\end{proof}

Now we tie the knots and prove the main lemma.

\begin{proof}[Proof of Lemma \ref{binary-positive-general}]
We first prove that the initial balance of group $i\in\{1,2\}$ at its first turn to pick a good is at least $h_i\cdot n_i$.
This is obvious for group 1 since, by definition of the initial payments, each member $j$ initially pays the group $B(r_j, s(r_j))$, which is at least $h_1$.
As for group 2, before its first turn it might have to pay to members who wanted the first good which was picked by group 1. 
To each member $j$, group 2 has to pay either 0 or $w(r_j,s(r_j))$. After than, the new net payment of each member is at least $B (r_j,s(r_j))-w(r_j,s(r_j)) = B(r_j-1,s(r_j))$, which is at least $h_2$.
	
By Lemma \ref{lem:balance-increasing}, the balance of each group weakly increases from its first turn to the end of the protocol, so the final balance of group $i$ is at least $h_i \cdot n_i$.
Lemma \ref{lem:rwav-end} then implies that
when RWAV ends, in each group $i$ there are at least $h_i \cdot n_i$ happy members.
\end{proof}

\begin{remark}
For simplicity, we assumed in Lemma \ref{binary-positive-general} that both groups have the same fairness criterion.
In general, however, each group $i$ can use a different function $s_i(\cdot)$ and get the corresponding guarantee regardless of the function $s_{-i}$ used by the other group. 
\end{remark}

\begin{remark}
It is possible to obtain improved guarantees if, instead of picking goods in deterministically alternating turns, we give the next turn to one of the two groups at  random. See Appendix~\ref{sec:randomized} for details.
\end{remark}

\subsection{1-of-best-$c$ fairness}
\label{sub:binary-1ofbestc}
We now use Lemma \ref{binary-positive-general} to obtain specific fairness guarantees, starting with 1-of-best-$c$ fairness.
Proposition \ref{binary-negative-1outofc} gives an upper bound of $h\leq 1-1/2^c$. The following theorem almost matches this bound.
\begin{theorem}
\label{binary-positive-1ofbestc}
For every $c\geq 2$,
RWAV can guarantee 
1-of-best-$c$ fairness 
to at least  $1-1/2^{c-1}$ of the members in both groups 
(and to at least $1-1/2^{c}$ of the members in the first group).
\end{theorem}
\begin{proof}
1-of-best-$c$ means that every agent with $r\geq c$ desired goods must get $s \geq 1$ goods.
By Lemma \ref{binary-positive-general}, we can guarantee this condition to at least $B(c-1,1)$ of the agents in both groups (and at least $B(c,1)$ of the agents in the first group). 
It remains to prove that $B(c,1) = 1 - 1/2^c$. The proof is by induction on $c$. For $c=0$, $B(0,1)=0$ by the boundary condition of $B$.
Now assume that $c>0$ and that the claim is true for $c-1$.
Then, by the recurrence \eqref{eq:brs} defining $B$:
\begin{align*}
B(c,1) &= 
\min\big[
(B(c-1,1)+B(c-1,0))/2,~~
B(c-2,0)
\big]
\\
&=
\min\big[
(1-1/2^{c-1} + 1)/2,~~
1
\big] \text{~~ (using the induction assumption)}
\\
&=
1-1/2^{c}.\qedhere
\end{align*}
\end{proof}
\begin{example}
\label{exm:binary-positive-1ofbestc}
Let $c=5$ and suppose that each agent in group 1 has 5 desired goods. 
By the proof of Theorem \ref{binary-positive-1ofbestc},
$w(r,1) = B(r-1)-B(r-1,1) = 1/2^r$.
When RWAV starts, all group members have $r=5$ and $s=1$ so their weight is $w(5,1)= 1/32.$ 
In the first turn, group 1 picks a good that is desired by some members; all these members now have $s=0$ so their weight changes to $0$ and they do not affect the voting from here on.
In the next turn, group 2 picks a good. If this good is desired by some members of group 1, then all these members have $r\to r-1 = 4$,
so their voting weight in the next turn becomes $1/2^4 = 1/16$---twice the voting weight of a member with $5$ desired goods. 
The process continues in this fashion, with
the weight of a member who ``loses'' a desired good (since it is taken by the other group)  multiplied by 2. Thus the interests of poorer agents are prioritized, in the spirit of the egalitarian philosophy.
\qed
\end{example}

Asymptotically the lower bound of Theorem \ref{binary-positive-1ofbestc}
almost matches the upper bound, but for small values of $c$ there is a gap. In particular, for $c=2$ the upper bound is $h \leq 2/3$ (Proposition \ref{binary-negative-2/3}) and the lower bound is $h\geq 1/2$. We conjecture that the correct value of $h$ is 2/3:
\begin{conjecture}
\label{conj:binary-positive-2/3}
$2/3$-democratic 1-of-best-2 fairness 
is always attainable.
\end{conjecture}

An equivalent formulation of the conjecture is that if every agent desires exactly two goods, then there exists an allocation such that at least two-thirds of the agents in each group receive at least one desirable good.

We support the conjecture with two theorems. The first improves the lower bound from $1/2$ to $3/5$ using an enhancement to the RWAV protocol.  The second improves the lower bound to $2/3$ for the special case of identical groups.

\begin{theorem}
\label{binary-positive-1ofbestc-improved}
For every $c\geq 2$, there is a protocol that guarantees
1-of-best-$c$ fairness to at least $\frac{2^c-1}{2^c+1}$ of the members in each group.
In particular, it guarantees $3/5$-democratic 1-of-best-$2$ fairness%
.
\end{theorem}
\begin{proof}
For 1-of-best-$c$ fairness, we can ignore all agents who want less than $c$ goods (since they may be given a value of 0), and assume that all agents want at least $c$ goods. Then the following protocol can be used.
\begin{framed}
\textbf{Enhanced RWAV protocol:}
\begin{itemize}
\item If, in one of the groups, at least $\frac{2^c-1}{2^c+1}$ of the agents desire the same good $g$, then give $g$ to that group and give all other goods to the other group.
\item Otherwise, run RWAV as usual.
\end{itemize}
\end{framed}

As in the proof of Lemma~\ref{binary-positive-general}, we have to prove that, for each group $i$, its balance when it first picks a good is at least $\frac{2^c-1}{2^c+1}\cdot n_i$.

We have $r\geq c$ and $s=1$ for all agents, so the initial payment of each agent is at least $B(c,1) = 1 - 1 / 2^c$.
Therefore the initial balance of group 1 is $\frac{2^c-1}{2^c}\cdot n_1 > \frac{2^c-1}{2^c+1}\cdot n_1$. 

As for group 2, before its first turn it might have to pay to members who ``lose'' a desired good to group 1. There are less than $\frac{2^c-1}{2^c+1}\cdot n_2$ such members, and the weight of each is at most $w(c,1)=1/2^c$. Therefore the initial balance of group 2 is above $(1-1/2^c) n_2 - (1/2^c)\cdot \frac{2^c-1}{2^c+1}\cdot n_2 = \frac{2^c-1}{2^c+1}\cdot n_2$.
\end{proof}

\begin{example}
\label{exm:binary-positive-1ofbestc-improved}
Consider an instance with five goods $\{v,w,x,y,z\}$ and two groups with 10 members in each family (for brevity, a member is represented by a concatenation of the goods he desires):
\begin{itemize}
\item Group 1: $vw,vx,vy,vz,wx,wy,wz,xy,xz,yz$;
\item Group 2: $vw,vw,vw,vx,vx,vx,vy,vy,vz,vz$.
\end{itemize}
Consider first RWAV without the enhancement. 
In group 1's first turn, all goods have the same total weight. 
For concreteness, we assume that ties are broken in the order $v>w>x>y>z$, so group 1 picks $v$.
In group 2's turn, $w$ and $x$ are tied, so it picks $w$. 
In group 1's next turn, all goods are again tied, so it picks $x$. Then group 2 picks $y$ and group 1 picks $z$.
In the final allocation, only 5 out of 10 members in group 2 ($vw,vw,vw,vy,vy$) are happy; this exactly matches the lower bound guaranteed by Theorem \ref{binary-positive-1ofbestc}.

In contrast, the enhanced RWAV gives $v$ to group 2 and all other goods to group 1. Then all members in both groups are happy.
\qed
\end{example}

Two groups are said to be \emph{identical} if 
there exists a bijection mapping each agent in one group to an agent in the other group with an identical utility function. We show that if the two groups are identical, the bound $2/3$ can be achieved.\footnote{We are grateful to Katie Edwards for the proof idea: https://math.stackexchange.com/a/2412319/29780.} This exactly matches the upper bound in Proposition~\ref{binary-negative-2/3}, which also applies to identical groups.
\begin{theorem}
\label{binary-positive-1ofbestc-identical}
For two identical groups,
there is an efficient protocol that guarantees 1-of-best-$2$ fairness to at least $2/3$ of the members in each group.
\end{theorem}
\begin{proof}
Assume without loss of generality that each agent desires exactly $2$ goods.
We use the following notation. Given an allocation $(G_1,G_2)$, for each good $g\in G$ and integer $y\in\{0,1,2\}$, define $p_y(g)$ (resp. $q_y(g)$) as the number of agents in group 1 (resp. 2) who want $g$ and receive utility $y$ from the allocation. Since the groups are identical, we have that for every $g\in G$, $p_0(g)=q_2(g)$, $p_1(g)=q_1(g)$, and $p_2(g)=q_0(g)$.

The protocol proceeds as follows. 
Start with an arbitrary allocation $(G_1,G_2)$.
If there is a good $g\in G_1$ for which $q_0(g) > p_1(g)$, move $g$ to $G_2$. 
Similarly, if there is a good $g\in G_2$ for which $p_0(g) > q_1(g)$, move $g$ to $G_1$. Stop when no good satisfies either of the conditions.

The number of utility-1 agents (across both groups) strictly increases in each iteration: when a good is 
taken from $G_1$, 
the number of utility-1 agents in it increases by $p_2(g)-p_1(g)$, and when it is given to $G_2$, the number of utility-1 agents in it increases by $q_0(g)-q_1(g)$, so the net increase is $p_2(g)+q_0(g) -p_1(g)-q_1(g) = 2 q_0(g) - 2 p_1(g) \geq 2$. Similarly, when a good is moved from $G_2$ to $G_1$ the number of utility-1 agents increases by $2 p_0(g) - 2 q_1(g) \geq 2$. Hence the algorithm stops after at most $(n_1+n_2)/2$ iterations.

When the algorithm stops, for all $g\in G_1$ we have $q_0(g)\leq q_1(g)$. Since each agent with utility 1 is counted once in $q_1(g)$
and each agent with utility 0 is counted twice in $q_0(g)$, this implies
that the number of utility-1 agents in group 2 is at least twice the number of utility-0 agents in group 2, so at most $1/3$ the agents in group 2 have utility 0. By similar considerations, at most $1/3$ the agents in group 1 have utility 0.
\end{proof}
\begin{example}
\label{exm:binary-positive-1ofbestc-identical}
Consider an instance with five goods $\{v,w,x,y,z\}$ and two identical groups with 10 members: $vw,vw,vw,vx,vx,vx,vy,vy,vz,vz$ 
(like group 2 in Example \ref{exm:binary-positive-1ofbestc-improved}).
Suppose the algorithm of Theorem \ref{binary-positive-1ofbestc-identical} starts by giving all goods to group 2, i.e., $G_1=\emptyset, G_2=\{v,w,x,y,z\}$.
The algorithm can proceed in many ways depending on the order in which the goods are considered.
For example, suppose $y$ is considered first. We have $p_0(y)=2$ while $q_1(y)=0$, so $y$ is moved to group 1.
Suppose $v$ is considered next. 
Now $p_0(v)=8$ while $q_1(v)=2$, so $v$ is moved to group 1.
The allocation becomes $G_1=\{v,y\}, G_2=\{w,x,z\}$.
Now $q_0(y)=2$ while $p_1(y)=0$, so $y$ is moved again to group 2. 
The final allocation is 
$G_1=\{v\}, G_2=\{w,x,y,z\}$.%
\footnote{
An implementation of this algorithm can be found at {https://github.com/erelsgl/family-fair-allocation}. It can be used to produce more examples.
}
\qed
\end{example}

\subsection{1-out-of-$c$ MMS fairness}
We now proceed to a stronger fairness requirement.
\begin{theorem}
\label{binary-positive-1-of-c}
For every $c\geq 3$,
RWAV can guarantee 
1-out-of-$c$ MMS-fairness 
to at least $1-1/2^{c-1}$ of the members in both groups.
\end{theorem}
\begin{proof}
1-out-of-$c$ MMS-fairness means that, for every $d\range{0}{c-1}$ and every $s\geq 1$, an agent with $r = c s + d$ desirable goods should receive at least $s$ such goods. 
Without loss of generality, we assume that all agents have only $c s$ desirable goods (if an agent has $c s + d$ desirable goods, we simply ignore $d$ of them). 
By Lemma \ref{binary-positive-general}, we can guarantee this fairness condition to at least $\inf_{s\geq 1} B(c s-1,s)$ of the agents in both groups. Hence, to prove the theorem it is sufficient to prove that, for every $c\geq 3$ and $s\geq 1$:
\begin{align*}
B(c s - 1, s) \geq 1-1/2^{c-1}.
\end{align*}
The proof requires various technical lemmas on binomial coefficients. These can be found in Appendix~\ref{sec:properties-of-B}. The claim itself is proved as Lemma \ref{binary-positive-1ofc}.\footnote{
This theorem appeared in the conference version of this paper \citep{ijcai2018democratic}.
Due to the error explained in Footnote~\ref{ftn:conference},
the proof there was much shorter. Happily, the theorem still holds after correcting the error, albeit with a much longer proof.
}
\end{proof}

While our main focus is on democratic fairness for any number of agents, we note that our results imply unanimous fairness when the number of agents is bounded.
\begin{corollary}
	\label{binary-unanimous}
	For two groups each containing at most $n$ agents with binary valuations, there exists a unanimous $1$-out-of-$(\lceil \log_2 (n+1) \rceil + 1)$ MMS-fair allocation.
\end{corollary}
\begin{proof}
	Let $c=\lceil \log_2 (n+1) \rceil + 1$, so $n\leq 2^{c-1}-1$. By Theorem \ref{binary-positive-1-of-c}, we can attain $(1-1/2^{c-1})$-democratic $1$-out-of-$c$ MMS-fairness. However, $(1-1/2^{c-1})$-democratic fairness and unanimous fairness are equivalent when the number of agents in each group is at most $2^{c-1}-1$. 
\end{proof}

Proposition \ref{binary-negative-1outofc} implies that the $O(\log n)$ rate cannot be improved.

\subsection{MMS and EF1}
Since Theorem \ref{binary-positive-1-of-c} is not valid for $c=2$, it does not say anything about MMS fairness (or about EF1, which is equivalent to MMS for binary agents). In fact, currently RWAV gives no meaningful lower bounds for MMS fairness. This is because the sequence $B(2 s-1, s)$ is decreasing and approaches 0 (see Table \ref{tab:B}). 
However, we can prove an existential lower bound that matches the upper bound of Proposition \ref{binary-negative-EFc}.\footnote{
In fact, in Section \ref{sec:2groups-additive} we present a protocol for two groups that guarantees $1/2$-democratic EF1 even for monotonic agents (Theorem~\ref{monotonic-positive-1/2}). 
However, Theorem \ref{binary-positive-ef1} is interesting since it shows that this guarantee can also be attained by a round-robin protocol.
}$^{,}$\footnote{We are grateful to J. Kreft for the proof idea: https://mathoverflow.net/a/307677/34461.}
\begin{theorem}
\label{binary-positive-ef1}
When a round-robin protocol is used for dividing goods between two groups: 

(a) The first group can always pick goods in such a way that the resulting allocation is envy-free for at least $1/2$ of its members; 

(b) The second group can always pick 
goods in such a way that the resulting allocation is EF1 (and MMS-fair) to at least $1/2$ of its members.
\end{theorem}
\begin{proof} We prove each part in turn.

\begin{enumerate}[(a)]
\item Each agent $j \in A_1$ with $2 s_j -1$ or $2 s_j$ desired goods needs to get at least $s_j$ of these goods in order to be envy-free. 
Without loss of generality, we assume that
$j$ has only $2 s_j - 1$ desirable goods (otherwise we ignore one such good arbitrarily).

We claim that group 1 has a strategy that guarantees to at least half its members $j$, at least $s_j$ desired goods.
Suppose by contradiction that the claim is false.
This means that, for \emph{every} picking strategy of group 1, group 2 can (adversarially) pick goods so that more than half of group 1's members get less than $s_j$ goods.
Equivalently, for more than half of group 1's members, at least $s_j$ of their desired goods are picked by group 2. 

However, if group 2 had such a strategy, group 1 could just copy this strategy and play it against group 2 (from the second step onwards). This would guarantee that more than half of the members of group 1 receive at least $s_j$ of their desired goods. 

\item We consider the game after group 1 picks its first good, as a new game in which group 2 plays first. By  part (a), group 2 has a strategy that guarantees that the allocation from this point on will be $1/2$-democratic EF. Accounting for the first good picked by group 1, the entire allocation is $1/2$-democratic EF1.
\qedhere
\end{enumerate}
\end{proof}

\subsection{Maximizing the fraction of happy agents}
While the main focus of this paper is on finding worst-case bounds on $h$, we briefly discuss the related problem of maximizing $h$ in a specific instance. 
In particular, in the spirit of egalitarianism, we would like to maximize the number of agents who get a positive fraction of their MMS.
Unfortunately this problem does not admit a polynomial-time algorithm unless $\text{P}=\text{NP}$: we prove this by showing that even deciding whether an instance admits an allocation that gives all agents a positive utility is NP-hard.
\begin{proposition}
\label{binary-npcomplete}
(a) Deciding whether a binary instance with two groups admits an allocation that gives every agent a positive utility is NP-complete.

(b) Deciding whether a binary instance with two groups admits an allocation that gives every agent a positive fraction of the MMS is NP-complete.
\end{proposition}

\begin{proof}
For any allocation, we can clearly verify in polynomial time whether it yields a positive utility to every agent, 
and whether it is positive-MMS fair for every agent.
We now show NP-hardness.

(a) We reduce from {\normalfont \scshape Monotone SAT}, a variant of the classical satisfiability problem where each clause contains either only positive literals or only negative literals. {\normalfont \scshape Monotone SAT} is known to be NP-hard \cite[p.~259]{GareyJo79}.
	
Given a {\normalfont \scshape Monotone SAT} formula $\phi$ with variables $x_1,\dots,x_m$, let there be $m$ goods corresponding to the $m$ variables. For each clause that contains only positive literals, we construct an agent in the first group who values exactly the goods contained in this clause. Similarly, for each clause that contains only negative literals, we construct an agent in the second group who values exactly the goods contained in this clause. Any assignment that satisfies $\phi$ gives rise to an allocation where the goods corresponding to true variables in the assignment are allocated to the first group and those corresponding to false variables in the assignment are allocated to the second group; this allocation gives every agent a positive utility (at least 1). Likewise, any allocation that gives every agent a positive utility yields a satisfying assignment of $\phi$. Hence the reduction is valid.
	
(b) We reduce from the problem of part (a), {\normalfont \scshape Positive Utility}. Given an instance, consider first all agents who desire a single good. If two of them are in different groups but desire the same good, then it is clearly impossible to give all agents a positive utility, so return False.
Otherwise, give all these agents their desired goods and remove them and their goods from the instance.
Repeat this process until there remains an instance in which every agent has at least two desirable goods, so the MMS of every agent is at least 1. An allocation is now positive-MMS for an agent if and only if it gives the agent a positive utility; 
use the solver for the positive-MMS problem and return its answer.
\end{proof}

The above reduction can be used in the opposite direction: if we have a SAT solver that can quickly solve {\normalfont \scshape Monotone SAT} problems, then we can use it to quickly decide whether a fair division instance admits a unanimous positive-MMS allocation.

However, this reduction does not work for the related maximization problem. Consider the problem MAX-SAT: given a formula $\phi$, find an assignment satisfying a maximum number of clauses in $\phi$. 
A solver for MAX-SAT can be used to find an allocation which gives a positive utility to a maximum number of agents, but ignores their groups. For example, it prefers an allocation in which 9 out of 10 members of the first group and 1 out of 10 members of the second group are happy, to an allocation in which 4 members of each groups are happy. 
In contrast, we are interested of maximizing the minimum number (or fraction) of happy agents in each group. For this we need to solve a problem that can be termed ``MAX-MIN-SAT'': given two formulas $\phi_1,\phi_2$, find an assignment which maximizes the minimum between the number of clauses satisfied in $\phi_1$ and number of clauses satisfied in $\phi_2$.
A more general problem, in which there are $k$ formulas, has been studied recently by \citet{bhangale2015simultaneous}. When $k$ is sufficiently small (in particular, when $k=2$ as in our case), they provide a polynomial-time constant-factor approximation.


\section{Two Groups with Additive or General Valuations}
\label{sec:2groups-additive}
In this section, we assume that there are two groups and each agent can have either an additive or a general monotonic utility function.

\subsection{Negative results}
The negative results for binary agents (Section \ref{sub:binary-negative}) obviously also hold for additive agents. In fact, for additive valuations we have a stronger impossibility: we cannot guarantee more than $1/2$ of the MMS to more than $1/3$ of the members in each group.
\begin{proposition}
\label{additive-negative-1/3}
For any $h>1/3$ and $q>1/2$, there is an additive instance with two groups in which no allocation is $h$-democratic $q$-fraction-MMS-fair.
\end{proposition}
\begin{proof}
	Consider an instance with $m=3$ goods and $n_1=n_2=3$ agents in each group, with utility vectors: $\textbf{u}_{i1}=(2,1,1)$, $\textbf{u}_{i2}=(1,2,1)$, and $\textbf{u}_{i3}=(1,1,2)$ for $i=1,2$. 
	The MMS of every agent is 2.
	In any allocation, one group receives at most one good, so at most one of its three agents receives utility more than 1.
	In that group, at most $1/3$ of the agents receive more than $1/2$ of their MMS.
\end{proof}
In the following subsection we will match this upper bound by showing that it is always possible to guarantee $1/2$ of the MMS to $1/2$ of the members in each group.

\subsection{Positive results}
The positive results for 1-of-best-$c$ fairness (Section \ref{sub:binary-1ofbestc}) are valid for additive valuations too: given an additive instance, we simply convert, for each agent, the utilities of her $c$ best goods to $1$ and the utilities of the other goods to $0$, and obtain a binary instance to which our previous positive results apply.

However, for the more general monotonic valuations, and for stronger fairness criteria such as EF1 and MMS, we need a different technique.

\begin{theorem}
\label{monotonic-positive-1/2}
For two groups of agents with monotonic valuations, 
$1/2$-democratic EF1 is attainable with an efficient protocol.
\end{theorem}
\begin{proof}
We arrange the goods on a line and process them from left to right. Starting from an empty block, we add one good at a time until the current block is EF1 for at least half of the agents in at least one group. 
We allocate the current block to one such group, and the remaining goods to the other group.

Since the whole set of goods is EF1 for both groups, the protocol terminates. Assume without loss of generality that the left block $G_1$ is allocated to the first group $A_1$, and the right block $G_2$ to the second group $A_2$. By the description of the protocol, the allocation is EF1 for at least half of the agents in $A_1$, so it remains to show that the same holds for $A_2$. Let $g$ be the last good added to the left block. More than half of the agents in $A_2$ think that $G_1\backslash\{g\}$ is not EF1, so for these agents, $G_1\backslash\{g\}$ is worth less than $G_2\cup\{g\}\backslash\{g'\}$ for any $g'\in G_2\cup\{g\}$. Taking $g'=g$, we find that these agents value $G_1\backslash\{g\}$ less than $G_2$. But this implies that the agents find $G_2$ to be EF1, completing the proof.
\end{proof}

An example run of this algorithm is shown in Appendix \ref{sample:line2}.

\begin{remark}
Proposition~\ref{binary-negative-EFc} shows that the factor $1/2$ in Theorem~\ref{monotonic-positive-1/2} cannot be improved even for binary agents and even if we relax EF1 to EF$c$ for any constant $c$. 
\end{remark}

\begin{remark}
Theorem~\ref{monotonic-positive-1/2} shows that if the goods lie on a line, we can find a $1/2$-democratic EF1 allocation that moreover gives each group a contiguous block on the line.
This may be important, for example, if the goods are houses on a street and each group wants to have all its houses in a contiguous block \citep{BarreraNyRu15,BouveretCeEl17,Suksompong17}.
\end{remark}

\begin{remark} 
The proof of Theorem~\ref{monotonic-positive-1/2} can be modified to show the existence of an allocation such that for at least half of the agents in each group, their envy towards the other group can be eliminated by removing the same good from the other group's bundle. This is similar to the notion of \emph{strong envy-freeness up to one good (s-EF1)} proposed for individual fair division by \cite{ConitzerFrSh19}, where the good that needs to be removed from an agent's bundle is independent of the envying agent.
\end{remark}

For additive agents, 
Theorem \ref{monotonic-positive-1/2} can be combined with Lemmas \ref{EF1-is-PROP} and \ref{PROP-is-MMS}:

\begin{corollary}
\label{additive-positive-1/2}
For two groups with additive agents, the following are attainable with an efficient protocol:

(a) $1/2$-democratic \prop{1};
	
(b) $1/2$-democratic $1/2$-fraction-MMS-fairness;
	
(c) $1/2$-democratic $1$-out-of-$3$ MMS-fairness;

(d) $1/2$-democratic MMS-fairness, if the valuations are binary.
\end{corollary}

The $1/2$-democratic  factor in part (b) is ``almost'' tight by Proposition \ref{additive-negative-1/3}. For  part (d), it is tight by Proposition \ref{binary-negative-EFc}.
For part (c) we do not know if it is tight: currently our best upper bound for $h$-democratic 1-out-of-3 fairness is $h\leq 7/8$ for binary agents, by Proposition \ref{binary-negative-1outofc}.

Note that part (d) was proved existentially in Theorem \ref{binary-positive-ef1}, but Corollary \ref{additive-positive-1/2} also gives  an efficient protocol.
~

While our main focus is on democratic fairness for any number of agents, 
we can attain a weaker unanimous fairness guarantee by reducing to cake cutting.
\begin{theorem}
	\label{additive-positive-unanimous}
	For any two groups of agents with additive valuations, there exists an allocation that is EF$(n-1)$ for all agents, where $n = n_1+n_2$ is the total number of agents in both groups.
\end{theorem}

\begin{proof}
	Choose an arbitrary agent in one of the groups. We will partition the goods into two parts and let the agent choose the part that she prefers. The resulting allocation is envy-free for this agent and hence also EF$(n-1)$ for her. It therefore suffices to show that there exists a partition in which each bundle is EF$(n-1)$ (with respect to the other bundle) for all of the remaining $n-1$ agents.
	
	To this end, assume that there is a divisible good (``cake'') represented by the half-open interval $(0,m]$. The value-density functions of the agents over the cake are piecewise-constant: for every $l\in\{1,\dots,m\}$, the value-density $v_{ij}$ in the half-open interval $(l-1,l]$ equals $u_{ij}(g_l)$.
	
	It is known that there exists a partition of the cake into two parts, using at most $n-1$ cuts, in which every agent has equal value for both parts \citep{Alon87}. Starting with two empty bundles, for each $l\in\{1,\dots,m\}$, we add good $g_l$ to the bundle corresponding to the part that contains at least half of the interval $(l-1,l]$. (If both parts contain exactly half of the interval, we add $g_l$ to an arbitrary bundle.) 
	
	We claim that every agent finds either bundle to be EF$(n-1)$. Fix an agent $a_{ij}$ and a bundle $G'$. From our partitioning choice, we have that $u_{ij}(G\backslash G')-u_{ij}(G')\leq u_{ij}(G'')$ for some set $G''\subseteq G\backslash G'$ of size at most $n-1$. This implies that the agent finds $G'$ to be EF$(n-1)$ with respect to $G\setminus G'$, as claimed.
\end{proof}
Note that the result of \citet{Alon87} is existential: as far as we know, there is no efficient way to compute the cake partition. Therefore the proof of Theorem \ref{additive-positive-unanimous} does not give rise to an efficient implementation.

\section{Three or More Groups}
\label{sec:3groups}

In this section, we study the most general setting where we allocate goods among any number of groups. Similarly to the previous sections, we start with negative results and move on to positive results. Recall that $k$ denotes the number of groups.

\subsection{Negative results}
The following proposition generalizes Proposition~\ref{binary-negative-2/3}.

\begin{proposition}
	\label{kgroups-binary-negative-2/3}
	For any $k\geq 2$ and any $h>k/(2k-1)$, there is a binary instance in which no allocation is $h$-democratic positive-MMS-fair.
\end{proposition}
\begin{proof}
	There are $2k-1$ goods placed in a circle. In each group there are $2k-1$ members. Each member of a group values a unique block of $k$ consecutive goods on the circle; the member has utility 1 for each of the $k$ goods and utility 0 for the remaining $k-1$ goods. Each agent has a positive MMS (1). However, in any allocation some group gets at most one good, and only $k$ members of the group get positive utility.
\end{proof}
In particular, when the number of groups is large, it is not possible to satisfy more than about half of the agents.

Next, we generalize Proposition \ref{binary-negative-maxh} to an arbitrary number of groups.
\begin{proposition}
\label{kgroups-binary-negative-maxh}
Let $r, s$ be integers such that $r\geq s \geq 1$.
There exists a binary instance with $k$ groups in which each agent desires exactly $r$ goods and needs $s$ goods in order to consider the allocation fair, where it is impossible to attain more than $\maxh(r,s)$-democratic fairness, where:
\begin{align*}
\maxh(r,s) = 
\begin{cases}
0 &  \text{~~when~~} r \leq k s - 1;
\\
\frac{1}{k^r}\sum_{i=s}^r (k-1)^{r-i} \binom{r}{i} & \text{~~when~~} r\geq k s.
\end{cases}
\end{align*}
\end{proposition}
\begin{proof}
For the case $r\leq k s - 1$, consider an instance with $r$ goods, where all agents in all groups desire all goods. At least one group will get at most $r/k<s$ goods, so all of its members will be unhappy.
	
For the case $r\geq k s$, consider an instance with $k m$ goods, for some $m\gg r$. In each group there are $\binom{k m}{r}$ members, each of whom wants a distinct subset of $r$ goods. At least one group will get at most $m$ goods. In this group, the fraction of happy agents will be at most:
	\begin{align*}
	\frac{
		\sum_{i=s}^r \binom{m}{i}\cdot \binom{km-m}{r-i}
	}{
		\binom{k m}{r}
	}.
	\end{align*}
	When $m\gg r$, the numerator is approximately $\sum_{i=s}^r \frac{m^i}{i!}\cdot  \frac{(k-1)^{r-i} m^{r-i}}{(r-i)!}
	=
	\sum_{i=s}^r\frac{ (k-1)^{r-i}  m^r}{i! (r-i)!}
	$ and the denominator is approximately $\frac{(k m)^r}{ r!} = \frac{k^r m^r}{r!}$. Therefore when $m\to\infty$ the expression approaches $\frac{1}{k^r}\sum_{i=s}^r (k-1)^{r-i}\binom{r}{i}$.
\end{proof}
In particular, we get the following generalization of Proposition \ref{binary-negative-1outofc}:
\begin{proposition}
	\label{kgroups-binary-negative-1outofc}
	For any integers $k,c\geq 2$ and any
	$h>1-\left(\frac{k-1}{k}\right)^c$, there is a binary instance with $k$ groups in which no allocation is $h$-democratic  1-of-best-$c$ fair (hence no allocation is 	$1$-out-of-$c$ MMS-fair).
\end{proposition}
\begin{proof}
	Apply Proposition \ref{kgroups-binary-negative-maxh} with $r=c$ and $s=1$.
	Then $\maxh(r,s) \leq \frac{1}{k^c} \sum_{i=1}^c (k-1)^{c-i}\binom{c}{i} = 
	\frac{k^c-(k-1)^c}{k^c} =
	1 - \left(\frac{k-1}{k}\right)^c$.
\end{proof}

As a generalization of Proposition \ref{binary-negative-EFc}, we get:
\begin{proposition}
	\label{kgroups-binary-negative-EFc}
	For any constant integer $c\geq 1$ and any $h>1/k$, there is a binary instance with $k$ groups in which no allocation is $h$-democratic EF$c$.
\end{proposition}
\begin{proof}
	Assume that there are $m=km'$ goods for some large positive integer $m'$. Each group consists of $2^m$ agents, each of whom values a distinct combination of the goods. Consider first an allocation that gives exactly $m'$ goods to each group, and fix a group. We claim that the fraction of the agents in the group whose utilities for some two bundles differ by at most $c$ converges to 0 for large $m'$. Indeed, this follows from the central limit theorem: Fix two bundles and consider a random agent from the group; let $X$ be the random variable denoting the (possibly negative) difference between the agent's utilities for the two bundles.  Then $X$ is a sum of $m'$ independent and identically distributed random variables with mean 0. The central limit theorem implies that for any fixed $\epsilon>0$, there exists a constant $d$ such that $\Pr[|X|\leq c]\leq \Pr[|X|\leq d\sqrt{m'}]\leq \epsilon$ for any sufficiently large $m'$. 
	
	Taking the union bound over all pairs of bundles, we find that the fraction of agents in the group who value some two bundles within $c$ of each other approaches 0 as $m'$ goes to infinity. This means that all but a negligible fraction of the agents find only one bundle to be EF$c$. By symmetry, $1/k$ of these agents find the bundle allocated to the group to be EF$c$. It follows that the fraction of agents in the group for whom the allocation is EF$c$ converges to $1/k$.
	
	It remains to consider the case where the allocation does not give the same number of goods to all groups. In this case, let $\mathcal{G}$ denote the set of bundles with the smallest number of goods, which must be strictly smaller than $m'$ goods. If we move goods from bundles with more than $m'$ goods to bundles in $\mathcal{G}$ in such a way that the number of goods in each bundle in $\mathcal{G}$ increases by exactly one, the fraction of agents in an arbitrary group that receives a bundle in $\mathcal{G}$ who finds the allocation to be EF$c$ can only increase. We can repeat this process, at each step possibly adding bundles to $\mathcal{G}$, until all bundles contain the same number of goods, which is the case we have already handled. Since the fraction of agents for whom the allocation is EF$c$ is bounded above by $1/k$ for large $m'$ in the latter allocation, and this fraction only increases during our process of moving goods, the same is true in the original allocation for the groups with less than $m'$ goods.
\end{proof}
~
\subsection{Approximate envy-freeness}
When there are two groups, the protocol in Theorem~\ref{monotonic-positive-1/2} is computationally efficient and yields an allocation that is both approximately envy-free and approximately proportional. In this subsection and the next, we present two ways of generalizing the result to multiple groups: one keeps the approximate envy-freeness guarantee but loses computational efficiency, while the other keeps only the approximate proportionality guarantee but also retains computational efficiency.

Our first theorem establishes the existence of a $1/k$-democratic EF2 allocation for $k$ groups with arbitrary monotonic valuations. Note that by Proposition~\ref{kgroups-binary-negative-EFc}, the factor $1/k$ cannot be improved.

\begin{theorem}
	\label{kgroups-monotonic-positive-EF}
	For $k$ groups with agents having arbitrary monotonic valuations, there exists an allocation that is $1/k$-democratic EF2.
\end{theorem}
\begin{proof}
	\citet{bilo2018almost} proved that for any $k$ agents with monotonic valuations, if the goods lie on a line, there exists an EF2 allocation that gives each agent a contiguous block on the line. We present their proof briefly (keeping only the details required for our purposes), and then show how to adapt their proof to the group setting.
	
	Bil\`{o} et al.'s protocol considers a $k$-vertex simplex in $\mathbb{R}^k$ that is triangulated to smaller $k$-vertex sub-simplices. 
	It identifies each vertex of the triangulation with a \emph{sub-partition}, i.e., a partition of some subset of $G$ into $k$ parts (some goods might not appear in any part). This identification has the following properties:
\begin{enumerate}
\item In each main vertex $i\in \{1,\ldots,k\}$ of the main simplex, part $i$ contains all of $G$ while the other $k-1$ parts are empty.
\item In each face spanned by main vertices $i_1,\ldots,i_l$ of the main simplex, only parts $i_1,\ldots,i_l$ contain goods, while the other $k-l$ parts are empty.
\item In each small sub-simplex of the triangulation, consider the $k$ sub-partitions attached to its $k$ vertices. For each good $g\in G$, there is a unique $i\in \{1,\ldots,k\}$ such that 
\begin{enumerate}[(a)]
\item $g$ belongs to part $i$ in at least one of these $k$ sub-partitions;
\item $g$ does not belong to another part in any of these $k$ sub-partitions.
\end{enumerate}
\item For each part $i\range{1}{k}$ in any such sub-partition, there are at most two goods that do not belong to this part $i$ but belong to part $i$ in some other sub-partition. 
\end{enumerate}
	For instance, if there are $k=4$ agents and $m=12$ goods, a possible set of four sub-partitions that satisfy properties 3 and 4 is the following:
	\begin{itemize}
		\item $(\{g_1,g_2,g_3\},\{g_4,g_5,g_6\},\{g_7,g_8,g_9\},\{g_{11},g_{12}\})$
		\item $(\{g_1,g_2,g_3\},\{g_5,g_6\},\{g_7,g_8,g_9\},\{g_{11},g_{12}\})$
		\item $(\{g_1,g_2,g_3\},\{g_5,g_6\},\{g_8,g_9\},\{g_{11},g_{12}\})$
		\item $(\{g_1,g_2,g_3\},\{g_5,g_6\},\{g_8,g_9,g_{10}\},\{g_{11},g_{12}\})$
	\end{itemize}
	
For each vertex of the triangulation, the protocol asks each agent which of the $k$ parts in the attached sub-partition he prefers the most, and labels the vertex with the answers. Since the agents' valuations are monotonic, we can assume that an agent never prefers an empty bundle to a non-empty bundle. Hence, by property 1, each main vertex $i$ only has label $i$. By property 2, all face vertices are labeled only with labels from the endpoints $i_1,\dots,i_l$ of the face. Thus, each agent's labeling satisfies Sperner's boundary condition. Therefore, by \citet{Bapat89}'s generalization of Sperner's lemma, there exists a sub-simplex and a matching of its vertices to the agents such that, in the vertex matched to agent $i$, agent $i$ prefers part $i$. By property 3, if we unite all parts numbered $i$ in all the $k$ sub-partitions of the sub-simplex, we get a partition of $G$. In the example above, the united partition is:
\begin{itemize}
\item $(\{g_1,g_2,g_3\},\{g_4,g_5,g_6\},\{g_7,g_8,g_9,g_{10}\},\{g_{11},g_{12}\})$.
\end{itemize} 
By property 4, in this united partition, each part $j$ is larger than the corresponding parts $j$ in the sub-partitions by at most two goods. Therefore, each agent $i$ finds part $i$ better than any other part $j$, up to at most two goods. 
	As a result, giving part $i$ of the united partition to agent $i$ yields an EF2 allocation.
	
	To adapt Bil\`{o} et al.'s proof to the group setting, we define a \emph{representative} for each group, and run their protocol on the $k$ representatives. 
	Whenever a representative is asked to select a best part, he uses \emph{plurality voting} among the group members, and answers by specifying the index of the part preferred by the largest number of members (breaking ties arbitrarily).
	Since each agent prefers a non-empty bundle to an empty bundle, the representative also prefers a non-empty bundle to an empty bundle. Hence the representative's answers satisfy Sperner's boundary conditions.
	This means that there exists a sub-simplex in which, in each vertex $i$, the representative of group $i$ prefers part $i$. 
	By the pigeonhole principle, in the sub-partition attached to this vertex $i$, at least $1/k$ of the members in group $i$ prefer part $i$. In the united partition, these members 
	find part $i$ better than any other part up to at most 2 goods.
	Therefore, giving part $i$ to group $i$ yields a $1/k$-democratic EF2 allocation.
\end{proof}

\begin{example}
Consider an instance with $k=4$ groups with 5 agents in each group. Consider the subsimplex whose four corners correspond to the four sub-partitions shown above. Suppose the outcomes of the plurality votes are as follows:
\begin{itemize}
\item
In allocation $(\{g_1,g_2,g_3\},\{g_4,g_5,g_6\},\{g_7,g_8,g_9\},\{g_{11},g_{12}\})$, in group 1: $\{g_1,g_2,g_3\}$ wins two votes while each of the other bundles wins one vote.
\item
In allocation $(\{g_1,g_2,g_3\},\{g_5,g_6\},\{g_7,g_8,g_9\},\{g_{11},g_{12}\})$, in group 2: $\{g_5,g_6\}$ wins two votes while each of the other bundles wins one vote.
\item
In allocation $(\{g_1,g_2,g_3\},\{g_5,g_6\},\{g_8,g_9\},\{g_{11},g_{12}\})$, in group 3: $\{g_8,g_9\}$ wins two votes while each of the other bundles wins one vote.
\item
In allocation $(\{g_1,g_2,g_3\},\{g_5,g_6\},\{g_8,g_9,g_{10}\},\{g_{11},g_{12}\})$, in group 4: $\{g_{11},g_{12}\}$ wins two votes while each of the other bundles wins one vote.
\end{itemize}
Then,  in the final allocation, group 1 gets $\{g_1,g_2,g_3\}$, group 2 gets $\{g_4,g_5,g_6\}$, group 3 gets $\{g_7,g_8,g_9,g_{10}\}$ and group 4 gets 
$\{g_{11},g_{12}\}$. In each group, at least two members believe that their group's bundle is better than all other bundles in one of the sub-partitions.  These same agents believe that their group's bundle is better-up-to-2-goods than all other bundles in the final allocation. Hence, the allocation is $2/5$-democratic EF2 (and therefore also $1/4$-democratic EF2). 
\qed
\end{example}

In the special case that all agents have \emph{binary} valuations, the fairness guarantee can be improved to EF1 by adapting the proof for the individual setting \citep{BarreraNyRu15,Suksompong17} to the group setting. The adaptation uses plurality voting similarly to  Theorem~\ref{kgroups-monotonic-positive-EF}. We only state the result here and defer the proof to Appendix~\ref{app:proof-kgroups}.

\begin{theorem}
	\label{kgroups-binary-positive-EF}
	For $k$ groups in which all agents have binary valuations, there exists an allocation that is $1/k$-democratic EF1, \prop{(k-1)} and MMS-fair.
\end{theorem}

By Proposition \ref{kgroups-binary-negative-EFc}, the factor $1/k$ is again tight.

\begin{remark}
	\citet{bilo2018almost} proved that, for $k\leq 4$ agents, the fairness guarantee can be improved from EF2 to EF1 even when the agents have arbitrary monotonic valuations. 
	It is possible that this result can be adapted to the group setting using plurality voting in a similar manner as in the previous two theorems. This would mean that for $k\leq 4$ groups of agents with arbitrary monotonic valuations, there exists a $1/k$-democratic EF1 allocation. However, their proof is rather involved and we have not been able to verify that our reduction works for it. 
	Likewise, for $k\geq 5$ groups we do not know whether $1/k$-democratic EF1 is attainable.
\end{remark}

\citet{BarreraNyRu15}, \citet{Suksompong17} and \citet{bilo2018almost} did not provide efficient algorithms for computing the corresponding approximate envy-free allocations. 
It is an interesting question whether such allocations can be found in polynomial time, both for the individual setting and the group setting.

\subsection{Approximate proportionality}
In this subsection, we show that if the fairness requirement is weakened from approximate envy-freeness to approximate proportionality, a $1/k$-democratic fair allocation can be attained in time polynomial in the input size.

\begin{theorem}
\label{kgroups-additive-positive-proportional}
When agents have additive valuations, it is possible to efficiently compute 
a $1/k$-democratic \prop{(k-1)} allocation.%
\footnote{
Note that \prop{(k-1)} is slightly stronger than the proportionality relaxations considered by  \cite{ConitzerFrSh17}, \cite{AzizCaIg18}, and \citet{Suksompong17}. 
By our Lemma \ref{PROP-is-MMS}, it also implies 
$1/k$-fraction-MMS-fairness, $1$-out-of-$(2 k-1)$ MMS-fairness, and (if all agents have binary utilities) also MMS-fairness.}
\end{theorem}
\begin{proof}
Similarly to the previous subsection, we use an existing algorithm for fair division among individuals \citep{Suksompong17} and adapt it to groups. We first describe the individual division algorithm.

We arrange the goods in a line and process them from left to right. Starting from an empty block, we add one good at a time. After each addition we ask each agent whether the current block is \prop{(k-1)}.
If one or more agents answer ``yes'', an arbitrary one among them receives the current block, and the process continues with the remaining $k-1$ agents. Finally, the remaining goods are given to the last agent. To prove that the allocation is fair, it is sufficient to prove that the remaining goods indeed satisfy \prop{(k-1)} for the last agent.

Number the agents $1,2,\dots,k$ according to the order in which they receive their bundle. For each $i\range{1}{k-1}$, 
denote the bundle already allocated to agent $i$ by $X_i\cup \{y_i\}$, where $y_i$ is the rightmost good (the last good added to the bundle) and $X_i$ is the set of the other goods in the bundle. By our assumptions, the bundles $X_1,\ldots,X_{k-1}$ do not satisfy \prop{(k-1)} for agent $k$.

The remaining bundle is $X_k = G\setminus\left( \cup_{i=1}^{k-1}X_i \bigcup \cup_{i=1}^{k-1}\{y_i\}\right)$. We claim that $X_k$ satisfies \prop{(k-1)} for agent $k$.
Specifically, we denote $C := \{y_1,\ldots,y_{k-1}\}$, note that $C\subseteq G\setminus X_k$, and claim that $u_k(X_k)\geq \frac{1}{k}u_k(G\setminus C)$.

Indeed, for all $i\range{1}{k-1}$: $C\subseteq G\setminus X_i$, and the bundle $X_i$ does not satisfy \prop{(k-1)}, so $u_k(X_i) < \frac{1}{k}u_k(G\setminus C)$. But $\cup_{i=1}^k X_i = G\setminus C$, so  
$\sum_{i=1}^k u_k(X_i) = u_k(G\setminus C)$.
Hence $u_k(X_k)>\frac{1}{k}u_k(G\setminus C)$.

We use the same algorithm for dividing the goods among groups.
Whenever a group is asked whether the current block satisfies \prop{(k-1)}, it poses this question to its members, and answers ``yes'' if at least $1/k$ of the agents answer yes.
To prove that the allocation is fair, it is sufficient to prove that the remaining goods satisfy \prop{(k-1)} for at least $1/k$ of the agents in the $k$th group. 
For each $i\range{1}{k-1}$, 
denote the bundle already allocated to group $i$ by $X_i\cup \{y_i\}$, where $y_i$ is the rightmost good and $X_i$ consists of the other goods. By our assumptions, the bundles $X_i$ satisfy \prop{(k-1)}
for less than $1/k$ of the members in group $k$. 
Hence, for more than $1/k$ of the members in group $k$, \emph{none} of the bundles $X_1,\ldots,X_{k-1}$ satisfies \prop{(k-1)}. For these members, $X_k$ must satisfy the fairness criterion due to the same argument as above. Hence, the allocation is $1/k$-democratic \prop{(k-1)}.
\end{proof}

An example run of this algorithm is shown in Appendix \ref{sample:line3}.


\subsection{1-of-best-$c$ fairness}
The results in the previous subsections guarantee $1/k$-democratic fairness, which means that the fraction of happy agents approaches 0 when the number of groups grows to infinity.
In this subsection we show that the 
RWAV protocol of Section \ref{sub:binary-positive} 
can be adapted to attain $h$-democratic 1-of-best-$c$ fairness where $h$ approaches a positive constant ($1/3$) even when $k\to \infty$.

\subsubsection{RWAV protocol for $k$ groups}
The round-robin protocol for $k$ groups proceeds exactly the same as for two groups: each group in its turn picks a good until all goods are taken.
As in Section \ref{sub:binary-positive}, each group picks a good using weighted approval voting among its members.
The only difference is in the weight-function $w(r,s)$.
For each $k\geq 2$, we use the following weight function (where $L_k := 2^{1/(k-1)}$):
\begin{align*}
w_k(r, 0) &= 0  \hspace{23.2mm} \forall r\in \mathbb{Z};
\\
w_k(r, 1) &= 
\begin{cases}
(L_k-1)/L_k^{r}  & r\geq 1;
\\
0  & r\leq 0.
\end{cases}
\end{align*}
Currently,  we know how to define the weights only for $s\leq 1$; hence we can handle only 1-of-best-$c$ fairness, for which $s\leq 1$.
Note that for $k=2$ we get $L_k=2$ and $w_k(r,1)=1/2^{r}$, which matches the calculation in Example \ref{exm:binary-positive-1ofbestc}.

\subsubsection{Analysis of RWAV}
Below we extend the analysis of Section \ref{sub:binary-positive} to $k$ groups. 
First, we extend the budget function $B(r,s)$ to:
\begin{align*}
B_k(r, 0) &= 1            \hspace{17.6mm} \forall r\in \mathbb{Z};
\\
B_k(r, 1) &= 
\begin{cases}
1 - 1/L_k^r  & r\geq 0;
\\
0  & r\leq 0.
\end{cases}
\end{align*}
Note that, analogously to the case $k=2$, we have:
\begin{align*}
w_k(r,s) = B_k(r,s) - B_k(r-1,s) && \forall s\in\{0,1\}, r\in\mathbb{Z}.
\end{align*}
We add to the protocol the following payment steps:
\begin{itemize}
\item \emph{Initialization}: each member $j$ of group $i$ pays $B_k(r_j,s_j)$ to group $i$.
\item Whenever group $i$ picks a good $g$, every member $j$ of group $i$ who wants $g$ pays $1-B_k(r_j,s_j)$ to group $i$.
\item Whenever a group $i'\neq i$ picks a good $g$, every member $j$ of group $i$ who wants $g$ receives $w_k(r_j,s_j)$ from group $i$.
\end{itemize}

The analogue of Lemma \ref{lem:balance} is:
\begin{lemma}
\label{lem:balance-k}
During the protocol, the balance of each agent $j$ is always $-B_k(r_j,s_j)$.
\end{lemma}
\begin{proof}
We prove the claim for members of group 1; the proof for other groups is analogous.
We proceed by induction. The induction base is handled by the initialization.

After a group $i'\neq 1$ picks a good $g$,
the balance of each member who wants $g$ increases.
The new  balance of each such member with $r$ remaining goods is:
\begin{align*}
- B_k(r,s) + w_k(r,s) ~~=~~-B_k(r-1,s)
\end{align*}
and indeed, for each such member, $r$ drops by 1 while $s$ does not change.

After group 1 picks a good $g$, 
the balance of each member who values $g$ at 1 becomes $-1 = -B_k(r_j,0)$, and indeed,  for each such member, $s$ becomes 0.
\end{proof}

Below we will need the following inequality on $w_k$:
\begin{align}
\label{eq:wk-series}
&
\sum_{i=2}^k w_k(r-i+2,1) = 
w_k(r,1) + \cdots + w_k(r-k+2,1)
\\
\notag
\leq ~ &
\frac{L_k-1}{L_k^r}(1 + \cdots + L_k^{k-2}) && \text{(definition of $w_k$)}
\\
\notag
= ~ &
\frac{L_k-1}{L_k^r}\cdot
\frac{L_k^{k-1}-1}{L_k-1}
 && \text{(sum of geometeric series)}
\\
\notag
= ~ &
\frac{2-1}{L_k^r} = \frac{1}{L_k^r}
 && \text{(definition of $L_k$)}
\\
\notag
= ~ &
1 - B_k(r,1)
 && \text{(definition of $B_k$)}
\end{align}
 (In fact, $w_k$ was calculated precisely to satisfy this inequality; it holds with equality when $r\geq k-1$.)

The analogue of Lemma \ref{lem:balance-increasing} is:
\begin{lemma}
\label{lem:balance-increasing-k}
For each group $i$, 
after $k$ consecutive turns starting at the turn of group $i$,
the balance of group $i$ weakly increases.
\end{lemma}
\begin{proof}
We prove the claim for group 1; the proof for other groups is analogous.
We calculate the change in the balance of group $1$ in a sequence of $k$ turns in which group $1$ picks a good $g_1$ and then the other groups pick goods $g_{2},\ldots,g_{k}$.
The members with $s_j=0$ have a weight of 0, so they neither pay nor receive anything; therefore we assume without loss of generality that all members have $s_j=1$.

Denote by $D_1$ the subset of group 1 members who desire $g_1$. Each member $j$ in $D_1$ pays to the group $1 - B_k(r_j,1)$, which by \eqref{eq:wk-series} above is at least:
\begin{align*}
\sum_{i=2}^k w_k(r-i+2,1).
\end{align*}


For $i\geq 2$, denote by $D_{i}$ the set of group 1 members who desire $g_i$.
Each member $j$ in $D_2\setminus D_1$ receives $w_k(r_j,1)$ from the group account.
Each member $j$ in $D_3\setminus D_1$ who had $r_j$ remaining goods at the start of the turn might now have only $r_j-1$ desired goods, so he may receive up to $\max[w_k(r_j-1,1),w_k(r_j,1)]$; since $w_k$ is a weakly decreasing function of $r$, this expression equals $w_k(r_j-1,1)$.
Similarly, for each $i\range{2}{k}$, each member in $D_i\setminus D_1$ may receive at most $w_k(r_j-i+2,1)$.
Therefore the total change in the group balance after the $k$ turns is:
\begin{align}
\notag
\Delta[Balance]
&
\geq
\sum_{j\in D_1}
\sum_{i=2}^k 
w_k(r_j-i+2,1)
-
\sum_{i=2}^k
\sum_{j\in D_{i}\setminus D_1}
w_k(r_j-i+2,1)
\\
\label{eq:DeltaBalance}
&
=
\sum_{i=2}^k 
\left(
\sum_{j\in D_1}
w_k(r_j-i+2,1)
-
\sum_{j\in D_{i}\setminus D_1} 
w_k(r_j-i+2,1)
\right).
\end{align}

Now, the group chose $g_1$ while $g_2,\ldots,g_k$ were still available, which means that the total weight of $g_1$ is weakly larger than the total weight of each of the other goods:
\begin{align*}
\forall i\range{2}{k}:
\sum_{j\in D_1} w_k(r_j,1)
\geq
\sum_{j\in D_{i}} w_k(r_j,1)
\geq
\sum_{j\in D_{i}\setminus D_1} w_k(r_j,1).
\end{align*}
By the expression for $w_k$, for every $i$, $w_k(r_j-i+2,1) = w_k(r_j,1)\cdot L_k^{i-2}$.
Therefore
$\sum w_k(r_j-i+2,1) = L_k^{i-2}\cdot \sum w_k(r_j,1)$, so the above inequality implies:
\begin{align*}
\forall i\range{2}{k}:
\sum_{j\in D_1} w_k(r_j-i+2,1)
\geq
\sum_{j\in D_{i}\setminus D_1} w_k(r_j-i+2,1).
\end{align*}
Substituting this in \eqref{eq:DeltaBalance} implies that $\Delta[Balance]\geq 0$.
\end{proof}

Lemma \ref{lem:rwav-end} is still true as-is in our case:
\emph{When the RWAV protocol ends, the balance of each group equals the number of its happy members.}
The proof is the same: when the protocol ends,
all happy members have $r=0, s=0$ so by Lemma \ref{lem:balance-k} their balance is $-B_k(0,0)=-1$, and all unhappy members have $r=0, s=1$ so by Lemma \ref{lem:balance-k} their balance is $-B_k(0,1)=0$. The group balance is the negative of the sum of its members' balances, which is exactly the number of happy members.

The analogue of Lemma \ref{binary-positive-general} is:
\begin{lemma}
\label{kgroups-binary-positive-general}
Let $c\geq 2$ be a constant. 
Suppose there are $k\geq 2$ groups with additive valuations. 
Then, the RWAV protocol 
yields an allocation that is 1-of-best-$c$ for at least 
a fraction $B_k(c-i+1,1)$ of the members in group $i$, for $i\range{1}{k}$.
\end{lemma}
\begin{proof}
Assume without loss of generality that each agent values his $c$ best goods at 1 and the other goods at 0. So for each agent we have $r=c$ and $s=1$.

For $i>c$, the value of $B_k(c-i+1,1)$ is non-positive and the statement holds trivially, so we may assume that $i\leq c$.
At the first turn of group $i$, each of its members has at least $c-i+1$ desired goods, so the group balance is at least $B_k(c-i+1,1)\cdot n_i$. 
From then on, by Lemma \ref{lem:balance-increasing-k}, the group balance weakly increases, so the final balance 
is at least $B_k(c-i+1,1)\cdot n_i$. 
By Lemma \ref{lem:rwav-end},
the final number of happy agents is at least $B_k(c-i+1,1)\cdot n_i$. 
\end{proof}

The analogue of Theorem \ref{binary-positive-1ofbestc} follows by substituting the expression for $B_k$:
\begin{theorem}
\label{kgroups-binary-positive-1ofbestc}
For every $c\geq k$,
RWAV can guarantee 
1-of-best-$c$ fairness 
to at least  $1-1/(L_k)^{c-k+1}$ of the members in all groups.
\end{theorem}
An interesting special case is when $c = k$, since 1-of-best-$k$ fairness is equivalent to positive-MMS fairness. 
Theorem \ref{kgroups-binary-positive-1ofbestc} guarantees this for $1 - 1/L_k = 1 - 2^{-1/(k-1)}$ of the agents; this ratio approaches $0$ for large $k$. However, by enhancing the RWAV protocol analogously to Theorem \ref{binary-positive-1ofbestc-improved}, we can guarantee 1-of-best-$k$ fairness to a constant fraction of the agents:
\begin{theorem}
\label{kgroups-positive-1ofbestk}
For $k$ groups with additive agents,
a $1/3$-democratic
1-of-best-$k$ allocation (which implies $1/3$-democratic positive-MMS) exists and can be found efficiently.
\end{theorem}
\begin{proof}
We convert all valuations to binary by assuming that each agent desires only his/her $k$ best goods (breaking ties arbitrarily).
We prove that it is possible to give at least $1/3$ of the agents in each group at least one desired good.
	
The proof is by induction on $k$. For $k=2$ we already proved in Theorem \ref{binary-positive-1ofbestc-improved}
that it is possible to guarantee 1-of-best-2 to at least $3/5$ of the agents in each group.
Assume the claim is true up until $k-1$; we will prove it for $k$.
	
If in some group $i$ at least $n_i/3$ members desire the same good $g$, give them good $g$ and divide the remaining goods among the remaining groups recursively. 
Note that in each remaining group, every agent now desires at least $k-1$ goods, so by the inductive hypothesis, it is possible to satisfy at least a $1/3$ fraction of each group.

Otherwise, run RWAV modified for $k$ groups as explained above. As in the proof of Lemma~\ref{kgroups-binary-positive-general}, 
it is sufficient to prove that, for each group $i$, its balance when it first picks an item is at least $n_i/3$.

The initial payment of each member is $B_k(k,1) = 1-1/2^{k/(k-1)} > 1/2$, so the initial amount paid to each group $i$ is more than  $n_i/2$. This is also the balance of group 1 when it first picks a good.

The balance of groups $2,3,\ldots,k$ is smaller since they have to pay to their members whose desired goods were picked. Obviously group $k$ is in the worst situation since it has to pay $k-1$ times, so we focus on this group. Each time a good is picked, the group has to pay to at most $n_k/3$ members. It has to pay $w_k(r,1)$ to each member with $r$ remaining goods. Recall that $w_k(r,1)$ is larger when $r$ is smaller. Therefore, the worst case for group $k$ is when it has to pay again and again to the same $n_k/3$ members. 
In this case it has to pay $n_k/3\cdot \sum_{r=2}^{k}w_k(r,1) =n_k/3\cdot [B_k(k,1)-B_k(1,1)]$. 
The total balance remaining in group $k$'s account when it first picks a good is thus at least:
\begin{align*}
B_k(k,1)\cdot n_k - n_k/3\cdot [B_k(k,1)-B_k(1,1)]
&>
B_k(k,1)\cdot 2n_k/3
\\
&>
(1/2)\cdot 2n_k/3 
= n_k/3.
\end{align*}
This concludes the proof.
\end{proof}

\begin{remark}
By Proposition \ref{kgroups-binary-negative-2/3}, the asymptotic upper bound on $h$ when $k\to\infty$ is~$1/2$.
\end{remark}

\section{Conclusion and Future Work}
For two groups, we have a comprehensive understanding of possible democratic fairness guarantees. 
We have a complete characterization of possible envy-freeness approximations, and upper and lower bounds for maximin-share-fairness approximations.
Some remaining gaps are shown in 
Table~\ref{table:summary}; closing them raises interesting combinatorial challenges.

For $k\geq 3$ groups, the challenges are much greater. 
Currently all our fairness guarantees are to no more than 1/3 of the agents in each group.
From a practical perspective, it may be important in some settings to give fairness guarantees to at least half of the agents in all groups. Finding protocols that provide such guarantees is an avenue for future work. 
In addition, our work leaves open the question of whether a stronger fairness notion than 1-of-best-$k$ can be guaranteed for a fraction of the agents in each group if we do not allow the fraction to decrease as the number of groups grows. 
From an algorithmic perspective, it is interesting whether there exists a polynomial-time algorithm that guarantees EF1 to any positive fraction of the agents. 

A possible concern about democratic fairness is that it completely leaves aside a fraction of the agents in each group. As Proposition~\ref{binary-negative-2/3} shows, it might be inevitable to leave some agents with zero utility. In these cases, the goal of an egalitarianist is to minimize the fraction of such poor agents. While the weighting scheme used by our RWAV protocol indeed prioritizes the interests of poor agents (see Example~\ref{exm:binary-positive-1ofbestc}), it may be interesting to develop an algorithm that directly minimizes the maximum fraction of poor agents across all groups.

We end the paper with a number of additional directions for future work.
\begin{itemize}
\item Our democratic fairness notion treats all groups equally regardless of their size. This may lead to situations where a large group has to sacrifice a significant amount of utility in order to preserve the fairness for a small group. How do the fairness guarantees change if we let the required fraction depend on the size of the group?
\item In several domains, including voting and preference elication, restricting the preferences of the agents is a common approach for circumventing negative results \citep{ElkindLaPe17}. Our group fair division setting opens up the possibility of imposing the same kind of conditions, for example by assuming that agents in the same group have single-peaked or single-crossing preferences. Is it possible to obtain stronger fairness guarantees if these conditions are satisfied?
\item We have not addressed the issue of efficiency in this paper beyond the assumption that all goods must be allocated. In individual fair division, it is known that EF1 and Pareto optimality are compatible \citep{CaragiannisKuMo16}. Can we similarly strengthen Theorem~\ref{monotonic-positive-1/2} by adding Pareto optimality? For binary valuations this is indeed possible since any Pareto improvement preserves EF1, but the question remains open for additive and general valuations.
\end{itemize}

\appendix
\section{Properties of the Function $B$}
\label{sec:properties-of-B}
The function $B(r,s)$, defined in Section \ref{sub:binary-positive}, represents a lower bound on the fraction of agents that can be given at least $s$ out of $r$ desired goods. It is defined using the following recurrence relation \eqref{eq:brs}:
\begin{align*}
\brs{r}{s} := 
\begin{cases}
1  &  s\leq 0;
\\
0  &  0<s \text{~and~} r<s;
\\
\min\bigg[
\frac{1}{2}[B(r-1,s)+B(r-1,s-1)]
,
B(r-2,s-1)
\bigg] & \text{otherwise}.
\end{cases}
\end{align*}
Some values are shown in Table \ref{tab:B}. In this section we prove several properties of $B$.

\begin{lemma}
	\label{B-decreasing}
	For every fixed $r$, $B(r,s)$ is a weakly decreasing function of $s$.
\end{lemma}
\begin{proof}
	By induction on $r$. For $r=0,1,2,3$ this is apparent from Table \ref{tab:B} (there are finitely many values to check).
	Now let $r\geq 4$. We assume the claim is true for $r-2$ and $r-1$ and prove it is true for $r$.
	
	$B(r,s)$ is a minimum of two expressions. In each of these expressions, the first operand is less than $r$. Therefore, 
	by the induction assumption, each of these expressions is decreasing with $s$.
	Therefore the same is true for $B(r,s)$.
\end{proof}

\begin{lemma}
	\label{B-increasing}
	For every fixed $s$, $B(r,s)$ is a weakly increasing function of $r$.
\end{lemma}
\begin{proof}
	We have to prove that, for every $s$ and every $r\geq 1$, $B(r,s)\geq B(r-1,s)$. 
	We prove this by induction on $r$. For $r=1,2,3$ this is apparent from Table \ref{tab:B} (there are finitely many values to check).
	Now let $r\geq 4$. We assume the claim is true for $r-2$ and $r-1$ and prove it is true for $r$. By the induction assumption, each term in the formula of $B(r,s)$ in \eqref{eq:brs} is no less than the corresponding term in the formula of $B(r-1,s)$. Hence it follows that $B(r,s)\geq B(r-1,s)$.
\end{proof}

\begin{lemma}
	\label{B-positive}
	For every $r,s$ such that $0\leq s\leq r$, 
	$0\leq B(r,s)\leq 1$.
\end{lemma}
\begin{proof}
	The boundary conditions on $B$ imply that, for every $r$, $B(r,0)=1$ and $B(r,r+1)=0$.
	Lemma \ref{B-decreasing} implies that, for every fixed $r$, $B(r,s)$ decreases from $1$ to $0$.
\end{proof}

\begin{lemma}
	\label{w-positive}
	For every $r,s$ such that $0\leq s\leq r$, 
	$0\leq w(r,s)\leq 1$.
\end{lemma}
\begin{proof}
By definition of $w$, $w(r,s) = B(r,s)-B(r-1,s)$.
By Lemma \ref{B-increasing}, this difference is at least 0.
By Lemma \ref{B-positive}, the difference is at most 1.
\end{proof}

\begin{lemma}
	\label{B-is-zero-when-r-is-small}
	For every $s\geq 1$ and every $r\leq 2 s-2$:
	\begin{align*}
	B(r,s)=0.
	\end{align*}
\end{lemma}
\begin{proof}
	By induction on $s$.
	When $s=1$, the claim should be verified only for $r=0$; indeed it is true by the boundary condition $B(0,1)=0$.
	
	Assume that $s\geq 2$ and that the claim is true for $s-1$.
	Let $r$ be an integer such that $r\leq 2 s - 2$.
	By the recurrence relation defining $B$:
	\begin{align*}
	B(r,s)\leq B(r-2,s-1).
	\end{align*}
	Since $r\leq 2 s - 2$, $r - 2 \leq 2 (s-1) - 2$.
	Hence by the induction assumption on $s-1$: $B(r-2,s-1)=0$,
	so $B(r,s)\leq 0$.
	But by Lemma \ref{B-positive}, $B(r,s)\geq 0$, so we must have $B(r,s)=0$.
\end{proof}

\begin{lemma}
	\label{large-r}
	For every $s\geq 1$ and every $r\geq 2 s -1$:
	\begin{align*}
	B(r,s) = \frac{1}{2}[B(r-1,s)+B(r-1,s-1)].
	\end{align*}
\end{lemma}
\begin{proof}
	By the recurrence relation \eqref{eq:brs}, it is sufficient to prove that whenever $r\geq 2 s-1$:
	\begin{align*}
	B(r-1,s)+B(r-1,s-1) - 2 B(r-2,s-1) \leq 0.
	\end{align*}
	
	We prove this by induction on $r$.
	For $r=1$ and $r=2$ we only have to check the case $s=1$; indeed the claim can be verified in Table~\ref{tab:B}.
	Assume that $r>2$ and that the claim holds for $r-1$ and $r-2$. We prove that it holds for $r$ by considering two cases.
	
	\emph{Case A}: $r = 2 s - 1$. Then, by Lemma \ref{B-is-zero-when-r-is-small}, $B(r-1,s)=0$.
	However, $B(r-1,s-1)$ and $B(r-2,s-1)$ are subject to the induction assumption, since $r-1 \geq 2(s-1)-1$ and $r-2 \geq 2(s-1)-1$. Hence:
	\begin{align*}
	&B(r-1,s)+B(r-1,s-1) - 2 B(r-2,s-1) 
	\\
	=\enspace&
	0 + \frac{1}{2}[B(r-2,s-1)+B(r-2,s-2)]
	- [B(r-3,s-1)+B(r-3,s-2)]
	\\
	=\enspace&
	\frac{1}{2}[B(r-2,s-1) + B(r-2,s-2) - 2 B(r-3,s-2)] ~~~~~ \text{(since $B(r-3,s-1)=0$)}
	\\
	\leq\enspace& 0 ~~~~~~~~ \text{(by the induction assumption on $r-1$, since $r-1\geq 2(s-1)-1$).}
	\end{align*}
	
	\emph{Case B}: $r \geq 2 s$. Then, all three terms in the inequality are subject to the induction assumption. Hence:
	\begin{align*}
	&B(r-1,s)+B(r-1,s-1) - 2 B(r-2,s-1) 
	\\
	=\enspace&
	\frac{1}{2}[
	B(r-2,s)+B(r-2,s-1)+
	B(r-2,s-1)+B(r-2,s-2)
	]
	\\
	&-[B(r-3,s-1)+B(r-3,s-2)]
	\\
	=\enspace&
	\frac{1}{2}[B(r-2,s) + B(r-2,s-1) - 2 B(r-3,s-1)]
	\\
	&+
	\frac{1}{2}[B(r-2,s-1) + B(r-2,s-2) - 2 B(r-3,s-2)]
	\\
	\leq\enspace& 0 ~~~~~~~~ \text{(by the induction assumption on $r-1$, since $r-1\geq 2s-1$).}\qedhere
	\end{align*}
\end{proof}

In light of Lemma \ref{large-r}, the recurrence relation  \eqref{eq:brs} can be simplified to:
\begin{align}
\label{eq:brs-simplified}
\brs{r}{s} := 
\begin{cases}
1  &  s\leq 0;
\\
0  &  
0<s \text{~and~} r\leq 2 s -2;
\\
\frac{1}{2}[B(r-1,s)+B(r-1,s-1)]
& 0<s \text{~and~}  r\geq 2 s -1.
\end{cases}
\end{align}

In the next lemma we find a closed-form solution to the function $B$ in \eqref{eq:brs-simplified}:

\begin{lemma}
\label{B-closed-form}
The following function satisfies the recurrence relation \eqref{eq:brs-simplified}:
\begin{align*}
B(r,s)=
\frac{1}{2^r}\sum_{i=s-1}^{r-s} \binom{r}{i}
=
\frac{1}{2^r}\sum_{i=s}^{r-s+1} \binom{r}{i},
\end{align*}
where we assume that $\binom{a}{b}=0$ if $b<0$ or $b>a$.
\end{lemma}
\begin{proof}
	When $s=0$, the sum goes from $-1$ to $r$, however, for $i=-1$ the summand is zero so $B(r,s) =\frac{1}{2^r}\sum_{i=0}^r \binom{r}{i} = 1$.
	
	When $r\leq 2 s -2$, the sum starts at $s-1$ and ends at (at most) $s-2$ so it is 0.
	
	When $r\geq 2 s -1$, it is sufficient to prove that $2^{r-1} B(r-1,s) + 2^{r-1} B(r-1,s-1) = 2^r B(r,s)$. We have
	\begin{align*}
	2^{r-1} B(r-1,s) &=  \sum_{i=s-1}^{r-s-1}\binom{r-1}{i}, \text{ and}
	\\
	2^{r-1} B(r-1,s-1) &= \sum_{i=s-2}^{r-s}\binom{r-1}{i}
	=\sum_{i=s-1}^{r-s+1}\binom{r-1}{i-1}
	\\
	&=
	\sum_{i=s-1}^{r-s-1}\binom{r-1}{i-1}
	+
	\bigg[
	\binom{r-1}{r-s-1}
	+\binom{r-1}{r-s} 
	\bigg],
	\end{align*}
	since when $r\geq 2 s -1$, the sum $\sum_{i=s-1}^{r-s+1}\binom{r-1}{i-1}$ contains at least two elements. Summing the above two equations gives:
	\begin{align*}
	&2^{r-1} [B(r-1,s) + B(r-1,s-1)]  \\
	&=  
	\sum_{i=s-1}^{r-s-1}\bigg[\binom{r-1}{i} + \binom{r-1}{i-1}\bigg] + 
	\bigg[
	\binom{r-1}{r-s-1}
	+\binom{r-1}{r-s}
	\bigg]
	\end{align*}
	By two applications of Pascal's identity:
	\begin{align*}
	2^{r-1}  [B(r-1,s) + B(r-1,s-1)]  &=  
	\sum_{i=s-1}^{r-s-1} \binom{r}{i} + 
	\binom{r}{r-s}
	\\
	&=
	\sum_{i=s-1}^{r-s}\binom{r}{i} = 2^r B(r,s).
	\qedhere
	\end{align*}
\end{proof}

Next, we prove several technical lemmas about binomial coefficients and their sums.
\begin{lemma}
	\label{3s-1}
	For every $s\geq 1$:
	\begin{align*}
	\binom{3s-1}{s-1}\frac{3s}{s+2}\leq 2^{3s-3}.
	\end{align*}
\end{lemma}
\begin{proof}
	By induction on $s$. For $s=1,2,3,4$ the claim can be verified manually. 
	We assume the claim for some $s\geq 4$ and prove it for $s+1$. When $s$ grows to $s+1$, the right-hand side is multiplied by $8$. The left-hand side is multiplied by:
	\begin{align*}
	\bigg[
	\binom{3s+2}{s}\frac{3s+3}{s+3}
	\bigg]
	\bigg/
	\bigg[
	\binom{3s-1}{s-1}\frac{3s}{s+2}
	\bigg]
	&=
	\frac{
		(3s+2)!(3s+3)(s-1)!(2s)!(s+2)
	}{
		(s)!(2s+2)!(s+3)(3s-1)!(3s)
	}
	\\
	&=
	\frac{
		(3s+2)(3s+1)(3s)(3s+3)(s+2)
	}{
		(s)(2s+2)(2s+1)(s+3)(3s)
	}
	\\
	&=
	\frac{
		3(3s+2)(3s+1)(s+2)
	}{
		2(s)(2s+1)(s+3)
	}
	\\
	&\leq
	\frac{
		3\cdot 3 \cdot(3s+2)
	}{
		2\cdot 2 \cdot (s)
	} = 2.25\cdot(3+2/s).
	\end{align*}
	When $s\geq 4$ this expression is less than 8, so the left-hand side remains smaller than the right-hand side.
\end{proof}

\begin{lemma}
	\label{3s-1*}
	For every $s\geq 1$:
	\begin{align*}
	\sum_{i=0}^{s - 1} \binom{3 s-1}{i}
	+
	\sum_{i=0}^{s - 2} \binom{3 s-1}{i}
	\leq 2^{3 s-3}.
	\end{align*}
\end{lemma}
\begin{proof}
	For $N\geq 2k$, denote by $f(N,k)$ the sum of the first $k$ binomial coefficients: $f(N,k) := \sum_{i=0}^k \binom{N}{k}$.
	Michael Lugo
	proved the following upper bound on this sum:\footnote{
		Here: 
		https://mathoverflow.net/a/17236/34461
	}
	\begin{align*}
	f(N,k)\leq \binom{N}{k}\frac{N-k+1}{N-2k+1}.
	\end{align*}
	Therefore:
	\begin{align*}
	f(N,k+1) + f(N,k) &=
	2 f(N,k+1) - \binom{N}{k+1}
	\\
	&\leq
	\binom{N}{k+1}
	\bigg[
	2 \cdot\frac{N-k}{N-2k-1}
	- 1
	\bigg]
	\\
	&=
	\binom{N}{k+1}
	\frac{N+1}{N-2k-1}.
	\end{align*}
	The left-hand side of the claim is this expression with $N=3s-1$ and $k=s-2$, so it is no more than:
	\begin{align*}
	\binom{3s-1}{s-1}
	\frac{3s}{3s-1-2s+4-1}
	=
	\binom{3s-1}{s-1}
	\frac{3s}{s+2},
	\end{align*}
	which by Lemma \ref{3s-1} is at most $2^{3s-3}$.
\end{proof}

Our next lemma is a generalization of Lemma \ref{3s-1*}.\footnote{
	We are grateful to Alex Francisco and Y. Forman for their help in proving this lemma here: 
	https://math.stackexchange.com/a/2604279/29780
}

\begin{lemma}
	\label{cs-1}
	For all integers $c\geq 3$ and $s\geq 1$:
	\begin{align*}
	\sum_{i=0}^{s-1} \binom{cs-1}{i}
	+
	\sum_{i=0}^{s-2} \binom{cs-1}{i}
	\leq 2^{cs-c}.
	\end{align*}
\end{lemma}
\begin{proof}
	We prove the claim by induction on $c$ for every fixed $s$. 
	For $c=3$, the inequality 
	follows from Lemma \ref{3s-1*}.
	We now assume the claim is true for some $c\geq 3$.  
	When $c$ grows to $c+1$, 
	the left-hand side still has the same number of summands ($2s-3$ summands), where in each summand, the $cs-1$ at the top becomes $cs+s-1$. 
	Meanwhile, 
	the right-hand side is multiplied by  $2^{s-1}$. Therefore, it is sufficient to show that in the left-hand side, each summand grows by a factor of at most $2^{s-1}$.
	Indeed, for every $i\leq s-1$:
	\begin{align*}
	\frac{
		\binom{cs+s-1}{i}
	}
	{
		\binom{cs-1}{i}
	}
	&=
	\frac{
		(cs+s-1)! / (cs+s-1-i)!
	}
	{
		(cs-1)! / (cs-1-i)!
	}
	\\
	&=
	\frac{
		(cs+s-1)\cdots(cs+s-i)
	}
	{
		(cs-1)\cdots(cs-i)
	}
	\\
	&=
	\bigg(
	1 + \frac{s}{cs-1}
	\bigg)
	\cdots
	\bigg(
	1 + \frac{s}{cs-i}
	\bigg)
	\\
	&\leq
	\bigg(
	1 + \frac{s}{cs-i}
	\bigg)^i
	&&\text{(the rightmost term is the largest)}
	\\
	&\leq
	\bigg(
	1 + \frac{s}{cs-(s-1)}
	\bigg)^{s-1}
	&&(i\leq s-1)
	\\
	&<
	\bigg(
	1 + \frac{s}{2s-s+1}
	\bigg)^{s-1}
	&&(c>2)
	\\
	&< 2^{s-1}.
	\end{align*}
	This completes the proof.
\end{proof}

We now use this combinatorial lemma to prove a useful lower bound on $B(r,s)$, which implies a democratic fairness guarantee.
\begin{lemma}
	\label{binary-positive-1ofc}
	For every $c\geq 3$ and $s\geq 2$:
	\begin{align*}
	B(c s - 1, s) \geq 1-1/2^{c-1}.
	\end{align*}
	Therefore,
	RWAV attains 
	$(1-1/2^{c-1})$-democratic 1-out-of-$c$ MMS-fairness.
\end{lemma}
\begin{proof}
	Using the closed form for $B(r,s)$ from Lemma \ref{B-closed-form}, we have to prove that:
	\begin{align*}
	&
	\frac{1}{2^{c s - 1}}\sum_{i=s}^{c s - s} \binom{cs-1}{i}
	\geq 1-\frac{1}{2^{c-1}}
	\\
	\iff
	&
	\frac{1}{2^{c s - 1}}
	\left[
	\sum_{i=0}^{s-1} \binom{cs-1}{i}
	+
	\sum_{i=cs-s+1}^{cs-1} \binom{cs-1}{i}
	\right]
	\leq \frac{1}{2^{c-1}}
	\\
	\iff
	&
	\sum_{i=0}^{s-1} \binom{cs-1}{i}
	+
	\sum_{i=cs-s+1}^{cs-1} \binom{cs-1}{i}
	\leq 2^{cs-c}
	\\
	\iff
	&
	\sum_{i=0}^{s-1} \binom{cs-1}{i}
	+
	\sum_{i=0}^{s-2} \binom{cs-1}{i}
	\leq 2^{cs-c},
	\end{align*}
	which we already proved in Lemma \ref{cs-1}.
\end{proof}

\section{Sample Runs of Some Protocols}
\label{sec:sample-runs}
Below we present sample runs of some of our allocation protocols. The source code used for the samples is available at: 
https://github.com/erelsgl/family-fair-allocation.

\subsection{RWAV protocol}
\label{sample:rwav}
Below are sample runs of the RWAV protocol (Section \ref{sub:binary-positive}) on an instance with five goods and two families with different fairness criteria. 

\begin{verbatim}
Group 1 seeks 1-out-of-2-maximin-share and has:
 * 2 binary agents who want ['v', 'x']
 * 1 binary agent  who want ['v', 'x', 'y']
 * 5 binary agents who want ['w', 'x', 'y', 'z']
 * 3 binary agents who want ['w', 'z']
Group 2 seeks one-of-best-2 and has:
 * 2 binary agents who want ['w', 'x', 'y', 'z']
 * 3 binary agents who want ['v', 'z']

-------

RWAV protocol - Group 1 plays first

Turn #1: Group 1's turn to pick a good from ['v', 'w', 'x', 'y', 'z']:
Calculating member weights:
            Desired set r  s  weight   
2 members   v,x         2  1  0.25     
1 member    v,x,y       3  1  0.125    
5 members   w,x,y,z     4  2  0.25     
3 members   w,z         2  1  0.25     
Calculating remaining good weights:
      Weight   
z     2.0      
v     0.625    
y     1.375    
x     1.875    
w     2.0      
Group 1 picks w

Turn #2: Group 2's turn to pick a good from ['v', 'x', 'y', 'z']:
Calculating member weights:
            Desired set r  s  weight   
2 members   w,x,y,z     3  1  0.125    
3 members   v,z         2  1  0.25     
Calculating remaining good weights:
      Weight   
z     1.0      
v     0.75     
y     0.25     
x     0.25     
Group 2 picks z

Turn #3: Group 1's turn to pick a good from ['v', 'x', 'y']:
Calculating member weights:
            Desired set r  s  weight   
2 members   v,x         2  1  0.25     
1 member    v,x,y       3  1  0.125    
5 members   w,x,y,z     2  1  0.25     
3 members   w,z         0  0  0        
Calculating remaining good weights:
      Weight   
v     0.625    
y     1.375    
x     1.875    
Group 1 picks x

Turn #4: Group 2's turn to pick a good from ['v', 'y']:
Calculating member weights:
            Desired set r  s  weight   
2 members   w,x,y,z     1  0  0        
3 members   v,z         1  0  0        
Calculating remaining good weights:
      Weight   
v     0        
y     0        
Group 2 picks v

Turn #5: Group 1's turn to pick a good from ['y']:
Calculating member weights:
            Desired set r  s  weight   
2 members   v,x         0  0  0        
1 member    v,x,y       1  0  0        
5 members   w,x,y,z     1  0  0        
3 members   w,z         0  0  0        
Calculating remaining good weights:
      Weight   
y     0        
Group 1 picks y

Final allocation:
 *  Group 1: allocated bundle = {'y', 'x', 'w'}, happy members = 11/11
 *  Group 2: allocated bundle = {'v', 'z'}, happy members = 5/5

-------

RWAV protocol - Group 2 plays first

Turn #1: Group 2's turn to pick a good from ['v', 'w', 'x', 'y', 'z']:
Calculating member weights:
            Desired set r  s  weight   
2 members   w,x,y,z     4  1  0.0625   
3 members   v,z         2  1  0.25     
Calculating remaining good weights:
      Weight   
z     0.875    
v     0.75     
y     0.125    
x     0.125    
w     0.125    
Group 2 picks z

Turn #2: Group 1's turn to pick a good from ['v', 'w', 'x', 'y']:
Calculating member weights:
            Desired set r  s  weight   
2 members   v,x         2  1  0.25     
1 member    v,x,y       3  1  0.125    
5 members   w,x,y,z     3  2  0.375    
3 members   w,z         1  1  0.5      
Calculating remaining good weights:
      Weight   
v     0.625    
y     2.0      
x     2.5      
w     3.375    
Group 1 picks w

Turn #3: Group 2's turn to pick a good from ['v', 'x', 'y']:
Calculating member weights:
            Desired set r  s  weight   
2 members   w,x,y,z     2  0  0        
3 members   v,z         1  0  0        
Calculating remaining good weights:
      Weight   
v     0        
y     0        
x     0        
Group 2 picks v

Turn #4: Group 1's turn to pick a good from ['x', 'y']:
Calculating member weights:
            Desired set r  s  weight   
2 members   v,x         1  1  0.5      
1 member    v,x,y       2  1  0.25     
5 members   w,x,y,z     2  1  0.25     
3 members   w,z         0  0  0        
Calculating remaining good weights:
      Weight   
y     1.5      
x     2.5      
Group 1 picks x

Turn #5: Group 2's turn to pick a good from ['y']:
Calculating member weights:
            Desired set r  s  weight   
2 members   w,x,y,z     1  0  0        
3 members   v,z         0  -1 0        
Calculating remaining good weights:
      Weight   
y     0        
Group 2 picks y

Final allocation:
 *  Group 2: allocated bundle = {'v', 'z', 'y'}, happy members = 5/5
 *  Group 1: allocated bundle = {'x', 'w'}, happy members = 11/11

\end{verbatim}

\subsection{Line-allocation protocol for two groups}
\label{sample:line2}
Below are three sample runs of the line-allocation algorithm of Theorem \ref{monotonic-positive-1/2}, on an instance with six goods and two families, where the fairness criterion is EF1.

\begin{verbatim}
Group 1 seeks envy-free-except-1 and has:
 * 7 agents with additive valuations: u=1 v=1 w=2 x=4 y=8 z=16
 * 2 agents with additive valuations: u=16 v=16 w=8 x=4 y=2 z=1
Group 2 seeks envy-free-except-1 and has:
 * 5 agents with additive valuations: u=1 v=1 w=1 x=3 y=3 z=4
 * 1 agent  with additive valuations: u=4 v=4 w=3 x=1 y=3 z=1
Group 3 seeks envy-free-except-1 and has:
 * 9 agents with additive valuations: u=1 v=1 w=1 x=2 y=3 z=3
 * 3 agents with additive valuations: u=3 v=3 w=3 x=2 y=1 z=1
 

----- Allocation between group 1 and group 2 -----

Current partition:  [] | ['u', 'v', 'w', 'x', 'y', 'z']:
   Group 1: 0/9 members think the left bundle is EF1
   Group 2: 0/6 members think the left bundle is EF1

Current partition:  ['u'] | ['v', 'w', 'x', 'y', 'z']:
   Group 1: 2/9 members think the left bundle is EF1
   Group 2: 0/6 members think the left bundle is EF1

Current partition:  ['u', 'v'] | ['w', 'x', 'y', 'z']:
   Group 1: 2/9 members think the left bundle is EF1
   Group 2: 1/6 members think the left bundle is EF1

Current partition:  ['u', 'v', 'w'] | ['x', 'y', 'z']:
   Group 1: 2/9 members think the left bundle is EF1
   Group 2: 1/6 members think the left bundle is EF1

Current partition:  ['u', 'v', 'w', 'x'] | ['y', 'z']:
   Group 1: 9/9 members think the left bundle is EF1
   Group 1 gets the left bundle
   Group 2 gets the remaining bundle

Final allocation:
 *  Group 1: allocated bundle = {'u', 'w', 'x', 'v'}, happy members = 9/9
 *  Group 2: allocated bundle = {'y', 'z'}, happy members = 5/6



----- Allocation between group 1 and group 3 -----

Current partition:  [] | ['u', 'v', 'w', 'x', 'y', 'z']:
   Group 1: 0/9 members think the left bundle is EF1
   Group 3: 0/12 members think the left bundle is EF1

Current partition:  ['u'] | ['v', 'w', 'x', 'y', 'z']:
   Group 1: 2/9 members think the left bundle is EF1
   Group 3: 0/12 members think the left bundle is EF1

Current partition:  ['u', 'v'] | ['w', 'x', 'y', 'z']:
   Group 1: 2/9 members think the left bundle is EF1
   Group 3: 3/12 members think the left bundle is EF1

Current partition:  ['u', 'v', 'w'] | ['x', 'y', 'z']:
   Group 1: 2/9 members think the left bundle is EF1
   Group 3: 3/12 members think the left bundle is EF1

Current partition:  ['u', 'v', 'w', 'x'] | ['y', 'z']:
   Group 1: 9/9 members think the left bundle is EF1
   Group 1 gets the left bundle
   Group 3 gets the remaining bundle

Final allocation:
 *  Group 1: allocated bundle = {'u', 'w', 'x', 'v'}, happy members = 9/9
 *  Group 3: allocated bundle = {'y', 'z'}, happy members = 9/12


----- Allocation between group 2 and group 3 -----

Current partition:  [] | ['z', 'y', 'x', 'w', 'v', 'u']:
   Group 2: 0/6 members think the left bundle is EF1
   Group 3: 0/12 members think the left bundle is EF1

Current partition:  ['z'] | ['y', 'x', 'w', 'v', 'u']:
   Group 2: 0/6 members think the left bundle is EF1
   Group 3: 0/12 members think the left bundle is EF1

Current partition:  ['z', 'y'] | ['x', 'w', 'v', 'u']:
   Group 2: 5/6 members think the left bundle is EF1
   Group 2 gets the left bundle
   Group 3 gets the remaining bundle

Final allocation:
 *  Group 2: allocated bundle = {'y', 'z'}, happy members = 5/6
 *  Group 3: allocated bundle = {'v', 'u', 'w', 'x'}, happy members = 12/12
\end{verbatim}

\subsection{Line-allocation protocol for three groups}
\label{sample:line3}
Below is a sample run of the line-allocation algorithm of Theorem 
\ref{kgroups-additive-positive-proportional}, on an instance with six goods and three families, where the fairness criterion is \prop{2}.

\begin{verbatim}
Group 1 seeks proportionality-except-2 and has:
 * 7 agents with additive valuations: u=1 v=1 w=2 x=4 y=8 z=16
 * 2 agents with additive valuations: u=16 v=16 w=8 x=4 y=2 z=1
Group 2 seeks proportionality-except-2 and has:
 * 5 agents with additive valuations: u=1 v=1 w=1 x=3 y=3 z=4
 * 1 agent  with additive valuations: u=4 v=4 w=3 x=1 y=3 z=1
Group 3 seeks proportionality-except-2 and has:
 * 9 agents with additive valuations: u=1 v=1 w=1 x=2 y=3 z=3
 * 3 agents with additive valuations: u=3 v=3 w=3 x=2 y=1 z=1

Current partition:  [] | ['u', 'v', 'w', 'x', 'y', 'z']:
   Group 1: 0/9 members think the left bundle is PROP-2
   Group 2: 0/6 members think the left bundle is PROP-2
   Group 3: 0/12 members think the left bundle is PROP-2

Current partition:  ['u'] | ['v', 'w', 'x', 'y', 'z']:
   Group 1: 2/9 members think the left bundle is PROP-2
   Group 2: 1/6 members think the left bundle is PROP-2
   Group 3: 3/12 members think the left bundle is PROP-2

Current partition:  ['u', 'v'] | ['w', 'x', 'y', 'z']:
   Group 1: 2/9 members think the left bundle is PROP-2
   Group 2: 6/6 members think the left bundle is PROP-2
   Group 2 gets the left bundle

Current partition:  [] | ['w', 'x', 'y', 'z']:
   Group 1: 0/9 members think the left bundle is PROP-2
   Group 3: 0/12 members think the left bundle is PROP-2

Current partition:  ['w'] | ['x', 'y', 'z']:
   Group 1: 2/9 members think the left bundle is PROP-2
   Group 3: 3/12 members think the left bundle is PROP-2

Current partition:  ['w', 'x'] | ['y', 'z']:
   Group 1: 9/9 members think the left bundle is PROP-2
   Group 1 gets the left bundle
   Group 3 gets the remaining bundle

Final allocation:
 *  Group 1: allocated bundle = {'x', 'w'}, happy members = 9/9
 *  Group 2: allocated bundle = {'v', 'u'}, happy members = 6/6
 *  Group 3: allocated bundle = {'y', 'z'}, happy members = 9/12
\end{verbatim}

\section{A Randomized Algorithm}
\label{sec:randomized}
While our main focus in this paper is on deterministic algorithms, it is interesting that we can obtain better democratic fairness guarantees by using a randomized algorithm. For simplicity, we illustrate this for two groups.

Instead of the RWAV protocol of Section \ref{sub:binary-positive},
we define a protocol called Coin-toss with Weighted Approval Voting (CWAV) as follows: 
\begin{framed}
\noindent
While there are remaining goods:
\begin{itemize}
\item Pick $i\in\{1,2\}$ uniformly at random.
\item Group $i$ picks a good.
\end{itemize}
\end{framed}

Each group picks its good using a weighted-approval scheme where the weights are defined by:

\begin{align*}
\crs{r}{s} &:= 
\begin{cases}
1  &  s\leq 0;
\\
0  &  0<s \text{~and~} r<s;
\\
\frac{1}{2}[\crs{r-1}{s}+\crs{r-1}{s-1}]
& \text{otherwise}.
\end{cases}
\\
w(r,s) &:= C(r,s) - C(r-1,s).
\end{align*}

CWAV is analyzed similarly to RWAV, by adding fiat payments. For simplicity we call the groups $1$ and $-1$.
\begin{itemize}
\item Initially, each member $j$ pays $C(r_j, s_j)$ to its group;
\item After group $i$ picks a good $g$, every member $j$ of $i$ who wants $g$ pays $w(r_j,s_j)$ to group $i$;
\item After group $-i$ picks a good $g$, every member $j$ of $i$ who wants $g$ receives $w(r_j,s_j)$ from group $i$.
\end{itemize}

Similarly to Lemma \ref{lem:balance}, 
it is easy to show that the balance of each agent $j$ is always $-C(r_j,s_j)$.
Instead of Lemma \ref{lem:balance-increasing}, we have:

\begin{lemma}
\label{lem:balance-increasing-C}
In each turn, 
the {expected} change in the 
balance of each group $i$ 
is weakly positive.
\end{lemma}
\begin{proof}
Suppose that, if group $i$ wins the coin-toss it picks a good $g_i$, while if group $-i$ wins it picks a good $g_{-i}$. The change in the balance of group $i$ is determined by its following subsets (not necessarily disjoint):
\begin{itemize}
\item $D_i$: members of group $i$ who desire $g_i$.
\item $D_{-i}$: members of group $i$ who desire $g_{-i}$.
\end{itemize}
With probability $1/2$, group $i$ wins the coin-toss, picks $g_i$, and receives $w(r_j,s_j)$ from each member in $D_i$.
With probability $1/2$, group $i$ loses the coin-toss and has to pay $w(r_j,s_j)$ to each member in $D_{-i}$.
Therefore the expected change in the group balance after one turn is:
\begin{align*}
\textbf{E}[\Delta[Balance]] =
\frac{1}{2}\bigg[
\sum_{j\in D_i} w(r_j,s_j) - \sum_{j\in D_{-i}} w(r_j,s_j)
\bigg].
\end{align*}
Since the group chose $g_i$ over $g_{-i}$, the total weight of $g_i$ is weakly larger, so the expected change in balance is $\geq 0$.
\end{proof}

Similarly to Lemma \ref{lem:rwav-end}, it is easy to show that, when the protocol ends, the group balance equals the number of its happy members.
Instead of Lemma \ref{binary-positive-general} we have:
\begin{lemma}
\label{binary-positive-randomized}
Given a fairness criterion represented by an integer function $s(r)$, 
the RWAV protocol 
yields an allocation that is fair for at least 
a fraction $h$ of the agents in each group, where:
\begin{align*}
h &= \inf_{r = 1, 2, \ldots} C(r,s(r)).
\end{align*}
\end{lemma}

\begin{proof}
The initial balance in each group $i$ is at least $h\cdot n_i$.
By Lemma \ref{lem:balance-increasing-C}, the expected value of the balance-increase after each coin-toss is weakly positive.
Since the coin-toss in each turn is independent of the other turns, the expected balance after the last turn is weakly larger than in the first turn.
The balance after the last turn equals the number of happy agents.
\end{proof}

Some values of $C(r,s)$ are shown in Table \ref{tab:C}. 
By solving the recurrence relation we can express $C(r,s)$ as:
\begin{align*}
C(r,s) = \frac{1}{2^r}\sum_{i=s}^r \binom{r}{i}.
\end{align*}
Note that this is the same as the function $\maxh$ of Proposition \ref{binary-negative-maxh}, except the part where $r\leq 2 s -1$, for which $\maxh=0$ but CWAV attains a positive fraction in expectation (Obviously, with a randomized protocol we can always attain an expected fraction of $1/2$ by simply giving all goods to a group chosen uniformly at random, so the range where $C(r,s)\leq 1/2$ is not interesting).

\begin{table}
	\begin{center}
		\STautoround{3}
		\begin{spreadtab}{{tabular}{>{\headingstyle}c|cccccccccc}}
			@ $r\downarrow ~ s \implies$
			& \textbf{:=0}        & \textbf{:={[-1,0]+1}}   & \textbf{:={[-1,0]+1}} & \textbf{:={[-1,0]+1}}  & \textbf{:={[-1,0]+1}}  & \textbf{:={[-1,0]+1}}  & \textbf{:={[-1,0]+1}}  & \textbf{:={[-1,0]+1}}  & \textbf{:={[-1,0]+1}}  & \textbf{:={[-1,0]+1}}
			\\ \hline
			0
			& 1          & 0                   & 0        & 0 & 0 & 0 & 0 & 0 & 0 & 0
			\\ 
			[0,-1]+1   
			& 1          & ([-1,-1]+[0,-1])/2   & ([-1,-1]+[0,-1])/2   & ([-1,-1]+[0,-1])/2   & ([-1,-1]+[0,-1])/2   & ([-1,-1]+[0,-1])/2   & ([-1,-1]+[0,-1])/2   & ([-1,-1]+[0,-1])/2   & ([-1,-1]+[0,-1])/2   & ([-1,-1]+[0,-1])/2   
			\\ 
			[0,-1]+1   
			& 1          & ([-1,-1]+[0,-1])/2   & ([-1,-1]+[0,-1])/2   & ([-1,-1]+[0,-1])/2   & ([-1,-1]+[0,-1])/2   & ([-1,-1]+[0,-1])/2   & ([-1,-1]+[0,-1])/2   & ([-1,-1]+[0,-1])/2   & ([-1,-1]+[0,-1])/2   & ([-1,-1]+[0,-1])/2   
			\\ 
			[0,-1]+1   
			& 1          & ([-1,-1]+[0,-1])/2   & ([-1,-1]+[0,-1])/2   & ([-1,-1]+[0,-1])/2   & ([-1,-1]+[0,-1])/2   & ([-1,-1]+[0,-1])/2   & ([-1,-1]+[0,-1])/2   & ([-1,-1]+[0,-1])/2   & ([-1,-1]+[0,-1])/2   & ([-1,-1]+[0,-1])/2   
			\\ 
			[0,-1]+1   
			& 1          & ([-1,-1]+[0,-1])/2   & ([-1,-1]+[0,-1])/2   & ([-1,-1]+[0,-1])/2   & ([-1,-1]+[0,-1])/2   & ([-1,-1]+[0,-1])/2   & ([-1,-1]+[0,-1])/2   & ([-1,-1]+[0,-1])/2   & ([-1,-1]+[0,-1])/2   & ([-1,-1]+[0,-1])/2   
			\\
			[0,-1]+1   
			& 1          & ([-1,-1]+[0,-1])/2   & ([-1,-1]+[0,-1])/2   & ([-1,-1]+[0,-1])/2   & ([-1,-1]+[0,-1])/2   & ([-1,-1]+[0,-1])/2   & ([-1,-1]+[0,-1])/2   & ([-1,-1]+[0,-1])/2   & ([-1,-1]+[0,-1])/2   & ([-1,-1]+[0,-1])/2   
			\\ 
			[0,-1]+1   
			& 1          & ([-1,-1]+[0,-1])/2   & ([-1,-1]+[0,-1])/2   & ([-1,-1]+[0,-1])/2   & ([-1,-1]+[0,-1])/2   & ([-1,-1]+[0,-1])/2   & ([-1,-1]+[0,-1])/2   & ([-1,-1]+[0,-1])/2   & ([-1,-1]+[0,-1])/2   & ([-1,-1]+[0,-1])/2   
			\\
			[0,-1]+1   
			& 1          & ([-1,-1]+[0,-1])/2   & ([-1,-1]+[0,-1])/2   & ([-1,-1]+[0,-1])/2   & ([-1,-1]+[0,-1])/2   & ([-1,-1]+[0,-1])/2   & ([-1,-1]+[0,-1])/2   & ([-1,-1]+[0,-1])/2   & ([-1,-1]+[0,-1])/2   & ([-1,-1]+[0,-1])/2   
			\\ 
			[0,-1]+1   
			& 1          & ([-1,-1]+[0,-1])/2   & ([-1,-1]+[0,-1])/2   & ([-1,-1]+[0,-1])/2   & ([-1,-1]+[0,-1])/2   & ([-1,-1]+[0,-1])/2   & ([-1,-1]+[0,-1])/2   & ([-1,-1]+[0,-1])/2   & ([-1,-1]+[0,-1])/2   & ([-1,-1]+[0,-1])/2   
			\\ 
			[0,-1]+1   
			& 1          & ([-1,-1]+[0,-1])/2   & ([-1,-1]+[0,-1])/2   & ([-1,-1]+[0,-1])/2   & ([-1,-1]+[0,-1])/2   & ([-1,-1]+[0,-1])/2   & ([-1,-1]+[0,-1])/2   & ([-1,-1]+[0,-1])/2   & ([-1,-1]+[0,-1])/2   & ([-1,-1]+[0,-1])/2   
			\\
			[0,-1]+1   
			& 1          & ([-1,-1]+[0,-1])/2   & ([-1,-1]+[0,-1])/2   & ([-1,-1]+[0,-1])/2   & ([-1,-1]+[0,-1])/2   & ([-1,-1]+[0,-1])/2   & ([-1,-1]+[0,-1])/2   & ([-1,-1]+[0,-1])/2   & ([-1,-1]+[0,-1])/2   & ([-1,-1]+[0,-1])/2   
			\\
		\end{spreadtab}
	\end{center}
	\caption{\label{tab:C}Some values of $C(r,s)$.
		Compare to $\maxh(r,s)$ in Table \ref{tab:maxh} and $B(r,s)$ in Table~\ref{tab:B}.
	}
\end{table}

\section{Proof of Theorem \ref{kgroups-binary-positive-EF}}
\label{app:proof-kgroups}
In this section we prove Theorem \ref{kgroups-binary-positive-EF}: \emph{For $k$ groups in which all agents have binary valuations, there exists an allocation that is $1/k$-democratic EF1, \prop{(k-1)} and MMS-fair.}

To establish this theorem, we prove two lemmas that may be of independent interest---one on cake-cutting and the other on group allocation for agents with additive valuations.

The result on cake-cutting 
generalizes the theorems of \citet{Stromquist80} and \citet{Su99}, who prove the existence of contiguous envy-free cake allocations for individual agents.
Since these results are well-known, we present the model and proof quite briefly, focusing on the changes required to generalize from individuals to groups.

We consider a ``cake'' modeled as the interval $[0,1]$. Each agent $a_{ij}$ has a value-density function $v_{ij}:[0,1]\to\mathbb{R}_{\geq 0}$. The value of an agent for a piece $X$ is $V_{ij}(X) = \int_{x\in X}v_{ij}(x)dx$.
Denoting by $X_i$ the allocation to group $i$, an allocation is \emph{envy-free} for an  agent $a_{ij}$ if $V_{ij}(X_i) \geq V_{ij}(X_{i'})$ for every group $i'$.
A \emph{contiguous allocation} is an allocation of the cake in which each group gets a contiguous interval.
\begin{lemma}
	\label{kgroups-positive-cake}
	There always exists a contiguous cake allocation that is $1/k$-democratic envy-free.
	The factor $1/k$ is tight.
\end{lemma}

\begin{proof}
	The space of all contiguous partitions  corresponds to the standard simplex in $\mathbb{R}^k$. Triangulate that simplex and assign each vertex of the triangulation to one of the groups. In each vertex, ask the group owning that vertex to select one of the $k$ pieces using \emph{plurality voting} among its members, breaking ties arbitrarily. Label that vertex with the group's selection.
	The resulting labeling satisfies the conditions of \emph{Sperner's lemma} (see \citet{Su99}). Therefore, the triangulation has a \emph{Sperner subsimplex}---a subsimplex all of whose labels are different. We can repeat this process with finer and finer triangulations. This gives an infinite sequence of smaller and smaller Sperner subsimplices. This sequence has a subsequence that converges to a single point. By the continuity of preferences, this limit point corresponds to a partition in which each group selects a different piece. Since the selection is by plurality, at least $1/k$ of the agents in each group prefer their group's piece over all other pieces.
	
	The tightness of the $1/k$ factor follows from Lemma 6 of \citet{SegalhaleviNi15}.
	It shows an example with $k$ groups and $n'$ agents in each group with the property that in order to give a positive value to $q$ out of $n'$ agents in each group, we need to cut the cake into at least $k (kq-n')/(k-1)$ intervals.
	In a contiguous partition there are exactly $k$ intervals. Therefore, the fraction of agents in each group that can be guaranteed a positive value is $q/n' \leq 1/k+1/n'-1/kn'$. Since $n'$ can be arbitrarily large, the largest fraction that can be guaranteed is $1/k$.
\end{proof}

The next lemma presents a reduction from approximate envy-free allocation of indivisible goods to envy-free cake-cutting. 
We call this approximation ``EF-minus-2''. 
For any agent $a_{ij}$, denote by $\uijmax:=\max_{g\in G}u_{ij}(g)$ the maximum utility of the agent for any single good.
An allocation is \emph{EF-minus-2} for agent $a_{ij}$ if for every group $i'$, $u_{ij}(G_i) > u_{ij}(G_{i'}) - 2 \uijmax$.
The reduction generalizes Theorem 5 of \citet{Suksompong17}; a similar reduction was used in Theorem 3 of \citet{BarreraNyRu15}.

\begin{lemma}
	\label{kgroups-additive-positive-EF}
	When agents have additive valuations, there always exists a contiguous allocation of indivisible goods that is $1/k$-democratic EF-minus-2.
\end{lemma}
\begin{proof}
	We create an instance of the cake-cutting problem in the following way.
	\begin{itemize}
		\item The cake is the half-open interval $(0,m]$.
		\item The value-density functions are  piecewise constant: 
		for every $l\in\{1,\ldots,m\}$,
		the value-density $v_{ij}$ in the half-open interval $(l-1,l]$ equals $u_{ij}(g_l)$.
	\end{itemize}
	
	By Lemma~\ref{kgroups-positive-cake}, there exists a contiguous cake allocation that is envy-free for at least $1/k$ of the agents in each group.
	From this allocation we construct an allocation of goods as follows.
	\begin{itemize}
		\item If point $g$ of the cake is in the interior of a piece, then good $g$ is given to the group owning that piece.
		\item If point $g$ of the cake is at the boundary between two pieces, then good $g$ is given to the group owning the piece to its left.
	\end{itemize}
	A group gets good $g$ only if it owns a positive fraction of the interval $(g-1,g]$.
	Hence, in the allocation, each group loses strictly less than the value of a good and gains strictly less than the value of a good (relative to its value in the cake division). 
	This means that every agent who believes that the cake allocation is envy-free also believes that the goods allocation is EF-minus-2.
\end{proof}

We are now ready to prove Theorem~\ref{kgroups-binary-positive-EF}.

\begin{proof}[Proof of Theorem~\ref{kgroups-binary-positive-EF}]
Suppose an allocation is EF-minus-2 for some agent $a_{ij}$. This means that the agent's envy towards any other group is less than $2 \uijmax \leq 2$. Since the agent has binary valuations, the envy is at most $1$, meaning that the allocation is EF1 for that agent. Hence any $1/k$-democratic EF-minus-2 allocation, which is guaranteed to exist by Lemma~\ref{kgroups-additive-positive-EF}, is also $1/k$-democratic EF1.
By Lemma \ref{EF1-is-PROP} it is also
$1/k$-democratic \prop{(k-1)}.
By Lemma \ref{PROP-is-MMS} it is also
$1/k$-democratic MMS-fair.
\end{proof}
The cake-cutting protocol of Lemma \ref{kgroups-positive-cake} might take infinitely many steps to converge. In fact, there is no finite protocol for contiguous envy-free cake-cutting even for individuals \citep{Stromquist2008Envyfree}. However, the division guaranteed by Lemma \ref{kgroups-additive-positive-EF} and Theorem \ref{kgroups-binary-positive-EF} can be found in finite time (exponential in the input size) by checking all possible allocations. An interesting open question is whether a faster algorithm exists.

\section*{Acknowledgments}
This work was partially supported by the European Research Council (ERC) under grant number 639945 (ACCORD) and by a Stanford Graduate Fellowship. Most of the work was done while the second author was a PhD student at Stanford University. We are grateful to 
Dan Hefetz,
Fedor Petrov, Arnaud Mortier, Darij Grinberg, Michael Korn, Kevin P. Costello, Nick Gill,  Jack D'Aurizio, Leon Bloy, Alex Francisco, Y. Forman, J. Kreft, Katie Edwards,  Yonatan Naamad, and the anonymous referees of IJCAI-ECAI 2018, COMSOC 2018 and Artificial Intelligence Journal for their helpful comments. 

\bibliographystyle{plainnat}
\bibliography{main-full}

\end{document}